\documentclass{article}

\usepackage[preprint]{neurips_2025}
\usepackage{adjustbox}

\usepackage[utf8]{inputenc} %
\usepackage[T1]{fontenc}    %
\usepackage{hyperref}       %
\usepackage{url}            %
\usepackage{booktabs}       %
\usepackage{amsfonts}       %
\usepackage{nicefrac}       %
\usepackage{microtype}      %
\usepackage{xcolor}         %

\usepackage{booktabs}
\usepackage{comment}
\usepackage{multirow}
\usepackage{amsmath}
\usepackage{amsfonts}
\usepackage{amssymb}
\usepackage{adjustbox}
\usepackage{graphicx}
\usepackage{subcaption}
\usepackage{wrapfig}
\usepackage{natbib}
\newtheorem{theorem}{Theorem}[section]
\newtheorem{corollary}{Corollary}[section]
\newtheorem{definition}{Definition}[section]
\newtheorem{lemma}[theorem]{Lemma}

\newtheorem{proof}[theorem]{Proof}

\title{Disentangling Locality and Entropy in \\ Ranking Distillation}

\author{Andrew Parry \quad Debasis Ganguly \quad Sean MacAvaney \\
School of Computing Science, University of Glasgow \\
\texttt{\small a.parry.1@research.gla.ac.uk, \{Debasis.Ganguly, Sean.MacAvaney\}@glasgow.ac.uk}
}

\begin{document}

\maketitle

\begin{abstract}

  The training process of ranking models involves two key data selection decisions: a sampling strategy (which selects the data to train on), and a labeling strategy (which provides the supervision signal over the sampled data). Modern ranking systems, especially those for performing semantic search, typically use a ``hard negative'' sampling strategy to identify challenging items using heuristics and a distillation labeling strategy to transfer ranking ``knowledge'' from a more capable model. In practice, these approaches have grown increasingly expensive and complex---for instance, popular pretrained rankers from SentenceTransformers involve 12 models in an ensemble with data provenance hampering reproducibility. Despite their complexity, modern sampling and labeling strategies have not been fully ablated, leaving the underlying source of effectiveness gains unclear.
  Thus, to better understand why models improve and potentially reduce the expense of training effective models, we conduct a broad ablation of sampling and distillation processes in neural ranking. We frame and theoretically derive the orthogonal nature of model geometry affected by example selection and the effect of teacher ranking entropy on ranking model optimization, establishing conditions in which data augmentation can effectively improve bias in a ranking model. Empirically, our investigation on established benchmarks and common architectures shows that sampling processes that were once highly effective in contrastive objectives may be spurious or harmful under distillation. We further investigate how data augmentation---in terms of inputs and targets---can affect effectiveness and the intrinsic behavior of models in ranking. Through this work, we aim to encourage more computationally efficient approaches that reduce focus on contrastive pairs and instead directly understand training dynamics under \textit{rankings}, which better represent real-world settings.
\end{abstract}

\section{Introduction}
Pre-trained language Models (PLMs)~\citep{vaswani:2023, devlin:2019} have been shown to be highly effective in ad-hoc ranking tasks~\citep{DBLP:series/synthesis/2021LinNY}. By training on large labeled datasets~\citep{nguyen:2016}, they can often outperform classical term-weighting models~\citep{DBLP:journals/corr/abs-1901-04085}. Since these early works, a key direction in improving PLM-based ranking models has been improving their data augmentation pipelines, which typically now combine ``hard'' negative mining~\citep{karpukhin:2020, qu:2021} (the deliberate selection of challenging non‑relevant texts), and distillation of relevance estimation from an existing teacher model~\citep{hinton:2015, lin:2020, hofstätter:2021}. While these techniques have independently demonstrated clear gains in representation learning~\citep{hsu:2021}, their interaction in ranking distillation settings is poorly understood and is applied in several works with little or no ablation~\citep{xiao:2022, ren:2023, song:2023, wang:2024}.

In Information Retrieval (IR), hard negatives are typically sampled from heuristic candidate sets scored by a preliminary model~\citep{karpukhin:2020, gao:2021}. In contrastive objectives, increasing the locality of (reducing the distance between) the candidate set creates a more challenging classification task~\citep{gutmann:2010, ceylan:2018}, often improving downstream effectiveness. Distillation, in turn, replaces binary labels with soft targets from a teacher~\citep{bucila:2006, ba:2013, hinton:2015}, allowing students to match or surpass larger estimators~\citep{pradeep:2022, xiao:2022, pradeep:2023}. Contemporary systems couple the two: teachers are employed both to score and to label documents, producing multistage cascades of models that must be trained, stored, and queried at scale.  Figure~\ref{fig:sampling} exemplifies this ossification, showing the five‑stage \textit{SentenceTransformers} pipeline whose twelve cross‑trained models are further filtered by a classifier~\citep{reimers:2019}, sometimes upon previous iterations of each other. Such complexity is problematic.  Each additional model inflates computational cost and CO$_2$ footprint~\citep{DBLP:conf/sigir/ScellsZZ22} and hinders reproducibility~\citep{DBLP:conf/sigir/WangMMO22}. From a theoretical standpoint, negative-selection heuristics influence only which instances are labelled, not the Lipschitz-geometry or teacher-entropy terms that influence generalization (cf. \ref{thm:main}).
\begin{wrapfigure}{r}{0.50\textwidth}  %
    \centering
    \vspace{-6pt}                      %
    \includegraphics[width=\linewidth]{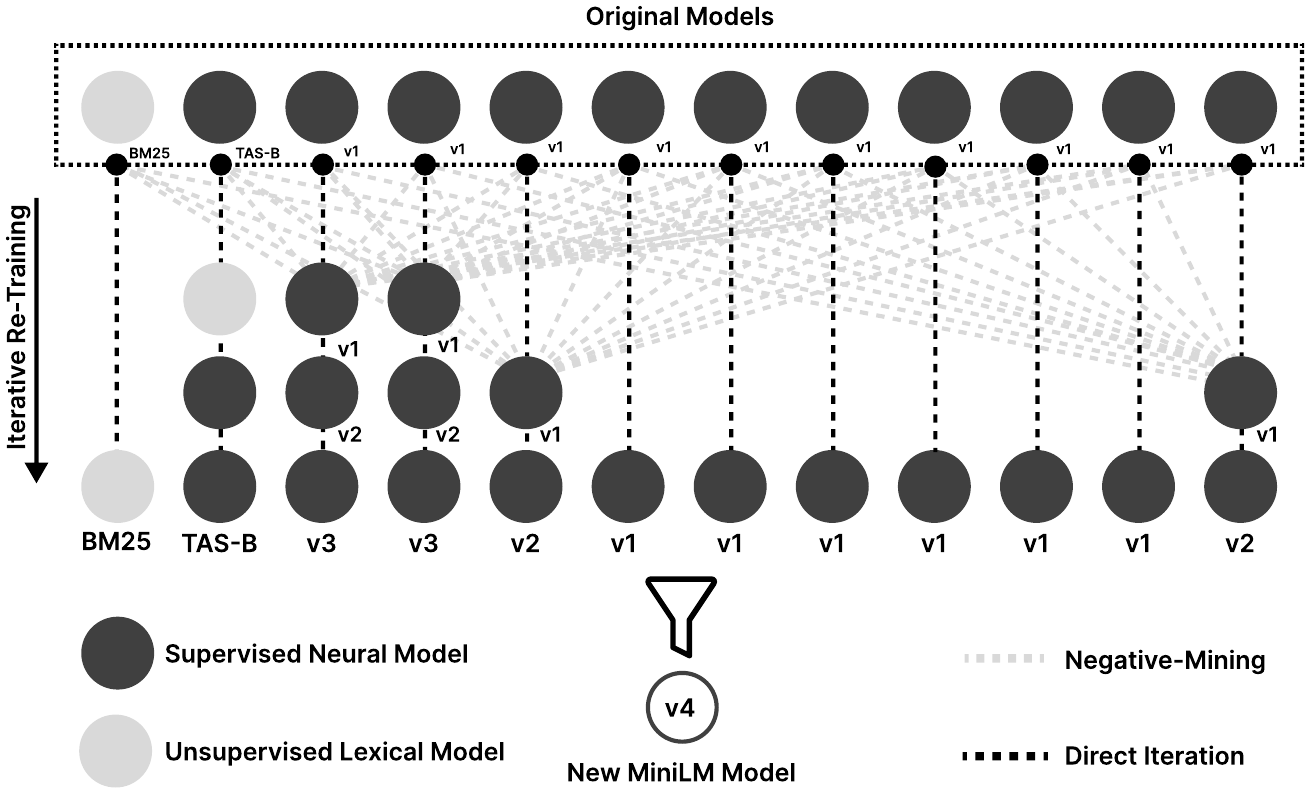}
    \caption{The existing SentenceTransformers~\citep{reimers:2019} pipeline for hard negatives. Even naive sampling pipelines can often have three stages of prior sampling~\citep{song:2023, wang:2024}. This practise reduces our ability to attribute effectiveness gain to a particular source and more broadly replicate this process due to cross-mining of examples from different stages in the pipeline.}
    \label{fig:sampling}
    \vspace{-6pt}                      %
\end{wrapfigure}
Furthermore, existing work in increasing the ``hardness'' of negatives often focuses solely on the domain of a ranking~\citep{karpukhin:2020, gao:2021, qu:2021}, either using expensive ensembling approaches as shown in Figure~\ref{fig:sampling} or iteratively choosing challenging domains based on heuristics such as model uncertainty~\citep{xiong:2020, zhan:2021}. Thus, much focus is on reducing epistemic uncertainty, which may lead to a highly confident model while neglecting the irreducible aleatoric uncertainty inherent in relevance labels. Such is the appeal of distillation that there may be multiple relevant documents within a single instance.

The notion of negative mining in ranking tasks is largely an artifact of contrastive objectives, often applied in representation learning. This nomenclature can be extended to functions such as ad-hoc search and reinforcement learning from human feedback~\citep{christiano:2017, rafailov:2024}; however, the increasing use of distillation or explicit annotation means that we operate over explicit rankings instead of solely positive-negative pairs. Thus, both theoretically and empirically, we investigate common training settings controlling for heuristics governing example locality and the entropy of target distributions. We underpin empirical ablations with a generalisation bound over ranking distillation (Sec. 3) whose bias term depends exclusively on (i) the intrinsic diameter of the query manifold and (ii) the teacher’s pairwise entropy—two quantities unchanged by additional negative-mining stages.

In isolating the contributions of ranking domains and target distributions in training settings, our primary contributions are two-fold: \textit{1) We show complex pipelines to be largely identical in effectiveness to naive approaches under semi-supervision. 2) We provide clear empirical evidence of the causal factors in model effectiveness when applying data augmentation.}

In investigating these factors, we aim to focus research on factors of effectiveness that are appropriate given a \textit{particular} search setting and reduce the spurious training and inference of multiple models.

\subsection{Related Work}

\noindent \textbf{Data Selection in Ranking}
Neural ranking models based on PLMs often apply contrastive objectives which benefit from the selection of similar instances, as these objectives become more challenging when examples are more similar in a geometric space but are different in terms of the target objective. \citet{gao:2021} argued that greater locality and a greater number of samples inspired by NCE~\citep{gutmann:2010} could further enhance the benefit of localized negatives, finding that negatives sampled from more precise rankings yielded greater effectiveness, which scales with the number of samples. Additionally, several approaches have been proposed to apply active learning to select negative examples which are most challenging to the model, which, while being quite effective, require a sufficient number of annotated queries to yield significant results~\citep{althammer:2023}. Several works propose that the flaw in negative mining is that as we produce an ever stronger source of negatives, due to label sparsity, we inevitably sample false negatives; as such, the notion of de-noised negatives has been proposed~\citep{qu:2021} and has become increasingly complex~\citep{reimers:2019} as illustrated in Figure \ref{fig:sampling}.

Distillation in some form has been applied both as weak supervision~\cite{dehghani:2017} and more commonly from a larger model in ranking~\citep{lin:2020, hofstätter:2021, qu:2021, xiao:2022}. For distillation minimal investigation has occurred to our knowledge to select an appropriate domain. Initial approaches adopted similar negative mining strategies those in contrastive settings~\citep{lin:2020, hofstätter:2021}, increasingly selection moves towards top-$k$ elements~\citep{pradeep:2023, schlatt:2025} as opposed to sampling due to both computational expense of large teacher models~\citep{sun:2023b} where a larger number of elements would not be ranked without being utilised and due to the ever increasing precision of PLM-based approaches~\citep{sun:2023}. For parity with contrastive objectives and analysis, we investigate the sampling of elements instead of top-$k$.

\noindent  \textbf{Understanding Distillation}
Knowledge distillation aims to transfer task effectiveness from one model to another. In ranking this approach has been applied extensively both as a weak-supervision signal for bootstrapping early attempts at PLM-based ranking and for efficiently training lightweight ranking models. Though some investigations provided qualitative explanations for the effectiveness of knowledge distillation~\citep{ba:2013, hinton:2015}, recent work has focused on more principled justifications of distillation and divergence-based learning objectives. The effectiveness of a teacher-student distillation setting have been examined in terms of intrinsic task difficulty~\citep{ji:2020} and the degree to which a student follows a teacher distribution~\citep{nagarajan:2019} suggesting that divergence from a teacher distribution is not a problem in model generalisation as shown empirically by methods in retrieval which ignore the exact densities of discrete ranking distributions~\citep{pradeep:2023, schlatt:2025}. In terms of explicit generalisation bounds, \citet{hsu:2021} provide a bound under uniform convergence, using distillation as a vector for understanding the original teacher model, we diverge from this setting as in downstream Information Retrieval we focus on trading off effectiveness for reduced latency. In terms of representation learning, the theoretical implications of data selection by an existing model have been explored~\citep{lin:2024} however this work does not extend to divergence-based losses explored in this work.

\section{Theoretical Analysis}
\label{sec:theory}

We now formalise the contribution of data augmentation to ranking model optimisation. We provide a generalisation bound in terms of sample locality and target entropy towards understanding where data selection can improve effectiveness.

\subsection{Preliminaries}
\noindent \textbf{The Ranking Problem} Given a corpus of texts, $\mathcal{C} = \{D_i\}_1^{|\mathcal{C}|}$ and a query $Q$, a top-$k$ ranking $\mathcal{R}=[D_i]_{1}^k$ (where $k \ll |\mathcal{C}|$) ordered by estimated relevance to $q$ determined by estimates some model $f: \mathcal{X} \rightarrow \mathbb{R}$ where $\mathcal{X}\equiv(Q,D)$. A learned ranking model is often modeled as an unnormalized estimator of $p(D | Q)$~\citep{robertson:1995},\footnote{$D$ may be one or more texts depending on architecture} modeling the likelihood of $D$ being relevant to $Q$. We focus on two common architectures, cross-encoders~\citep{nogueira:2020} and bi-encoders~\citep{karpukhin:2020}. A cross-encoder treats ranking as a regression over the joint representation of a query and document encoded by a PLM from which a relevance score is estimated. A bi-encoder instead separately encodes queries and documents, treating ranking as a maximum inner-product search problem over pooled latent representations.

\noindent \textbf{Training Ranking Models} Let $\mathcal{Q}$ denote a set of training queries. For each $Q \in \mathcal{Q}$ we observe a finite candidate list $\mathcal{D}_Q$ commonly treated as pairs $\mathcal{X}_Q = \{(Q, D_i): \forall D_i \in \mathcal{D}_Q\}$. Each element of $\mathcal{X}_Q$ can be assigned a binary relevance label $\mathcal{Y}_Q=\{y_i\}_1^{|\mathcal{X}_Q|}, y\in \{0,1\}$. Commonly, solely `positive' candidates are explicitly annotated (i.e., $y \gets 1$)~\citep{nguyen:2016}, forming a labelled set $\mathcal{L}$. All other elements are sampled from a larger set of similar documents chosen by a heuristic, forming a pseudo-negative set $\mathcal{U}$. In semi-supervised settings, targets $\mathcal{Y}_Q$ are instead determined by an existing teacher model $g$ such that $\mathcal{Y}_Q=\{g(Q, D), \forall (Q,D) \in \mathcal{X}_Q\}$.

\noindent \textbf{Choosing Pairs} It is common to condition the sampling of pseudo-negative examples on some existing scorer~\citep{karpukhin:2020, qu:2021}. Formally, let $\mu_Q$ define a query-specific measure over $\mathcal{X}$, and let $\nu_Q$ denote a biased measure derived from some existing model for data selection. Recent approaches use an ensemble of systems~\citep{qu:2021, song:2023} inducing $\nu_Q$, which, in contrastive settings, can be further filtered by a model $g$ to ensure that $d(\mathcal{X}_i, \mathcal{X}_j) > \epsilon, \forall \mathcal{X}_{j, i \neq j} \in \mathcal{X}$. This assumes that $\mathcal{X}_i$ can be a labeled positive for the contrastive objective. This step, given its expense and implied confidence in the discriminative ability of $g$, would imply that one should learn an approximation of the teacher model $g$ in a semi-supervised setting as opposed to solely filtering a candidate distribution. Nevertheless, this approach of ensembling and filtering is frequently employed, and thus we consider it for completeness.

\noindent \textbf{Problem Setting} The subjectivity of relevance and the scale of modern web-scale corpora make estimation of a ranking \emph{intrinsically subjective}~\citep{voorhees:1998, parry:2025}:  for most queries, we neither observe a complete ordering of candidates nor can we assume perfect recall in finding relevant documents, thus there can be many relevant documents beyond the single-relevant document constraint of contrastive learning. We therefore often work with \textit{a)} a sparse set of human judgements and noisy negative examples, or \textit{b)} a teacher model $g:\mathcal X\!\to\!\mathbb R$ trained on auxiliary data. The objective of student $f: \mathcal{X} \to \mathbb{R}$ is to rank elements from a structured space $\mathcal{X}$ equipped with metric $d$. Given a query $Q$, the goal is to learn a scoring function that produces rankings aligned with a teacher model $g$.

Recall that for each query $Q\in \mathcal{Q}$ we have a candidate pairs $\mathcal{X}_Q=\{(Q, D_{1}),\dots,(Q,D_{m_Q})\}$ together with labels $\mathcal{Y}_{Q}\in\{y_0,y_1,\dots,y_{m_Q}\}\cup\{\varnothing\}$. A student ranker $f:\mathcal X\!\to\!\mathbb R$ outputs real‑valued scores whose descending order defines the predicted ranking. Our loss criteria will contrast query-document pairs such that we look to minimise the pair-wise risk of misordering two pairs compared to their order under model $g$. Metric structure matters fundamentally because our student learns in a setting where available training data represents a limited and often biased sample from the true space of relevant elements around an anchor query. The geometry constrains how knowledge can transfer between observed pairs.
\label{subsec:problem}

Thus, our problem lies in the contribution of data augmentation to the effectiveness of ranking models. Distillation through criteria such as RankNet~\citep{burges:2010} blurs the boundary between classical notions of contrastive learning and knowledge distillation. No score values are used, and technically, contrastive objectives do the same, albeit with arbitrary negative ordering, both conditioned on some teacher (or negative miner) $g$. Negative sampling and knowledge distillation can be seen as orthogonal; sampling provides suitable observations given a downstream task and requires some labeling process for optimization; distillation provides the labeling process for optimization, but does not provide an explicit selection process for observations.

\subsection{Notation and Definitions}
\label{subsec:entropy}

Assuming student model $f$ and teacher model $g$, we use the output of model $g$ as targets $\mathcal{Y}$. We define locality in terms of a query-specific measure $\mu_Q$ over our input space $\mathcal{X}$ modelling the geometry of relevant elements conditioned on $q$. Formally, let $(\mathcal X,\mu_Q,d)$ be a metric–measure space with complete space $\mathcal{X}$, measure $\mu_Q$ conditional on $Q$ and distance $d$ (we use cosine distance over latent representations), define the essential diameter of this space as
\[\Delta_Q = \operatorname*{ess\,sup}_{(x,x') \in \mathcal{X}^2} \ d(x,x') \text{ with respect to } \mu_Q \otimes \mu_Q
\tag{1}\label{eq:diameter}
\] 

Hypothesising that due to the higher entropy of ranking targets encode additional useful information, similarly to the propositions of \citet{hinton:2015}, we model the entropy of a teacher ranking under pair-wise preferences. We investigate losses that can be seen as Bregman divergences as they are prevalent in the empirical ranking literature (RankNet~\citep{burges:2010}, MarginMSE~\citep{hofstätter:2021}, KL-divergence~\citep{kullback:1951}) and admit clean theoretical analysis. Define the pair-wise risk of $f$ as:

\begin{definition}[Pair-wise Risk]
For a scorer $f$ and query measure $\mu_Q$, the \emph{pair-wise risk} is the probability of mis-ordering:
\[\mathcal{R}(f) = \Pr_{(x,x') \sim \mu_Q^{\otimes 2}}[f(x) < f(x')].
\tag{2}
\]
\end{definition}

To minimise this risk, we use a distillation loss of the form:

\begin{definition}[Bregman Distillation Loss]
For convex potential $\phi$, the distillation loss on pair $(x,x')$ is
\[\ell(f,g;x,x') = D_\phi\left(f(x)-f(x') \,\|\, g(x)-g(x')\right),
\tag{3}
\]
where $D_\phi(a\|b) = \phi(a) - \phi(b) - \phi'(b)(a-b)$ is the Bregman divergence.
\end{definition}

Under this pair-wise setting we define:
\begin{definition}[Query Entropy]
For teacher $g$ and query measure $\mu_Q$, define the \emph{query entropy} as
\[H(g) = \mathbb{E}_{(x,x') \sim \mu_Q^{\otimes 2}}\left[H\left(\mathbf{1}[g(x) > g(x')]\right)\right],
\tag{4}
\]
where $H(p) = -p\log p - (1-p)\log(1-p)$ is binary entropy.
\end{definition}

When entropy is too low, this can be seen as a trivial setting under a Bregman divergence criterion, empirically leading to collapse as the optimisation task is too easy. When elements are more difficult to distinguish as entropy increases, we look to understand how the domain over which a ranking is calculated affects optimisation in settings where entropy is high. To measure the contribution of this entropy to the optimisation of a student, we define the misordering probability of model $g$ via application of Pinsker's inequality:

\begin{definition}[Misordering-Probability under Entropy (Proof in Appendix \ref{app:pinsker-entropy})]
For teacher $g$ with pair-wise entropy $H(g)$, using surrogate conversion following \citet{painsky:2020} the misordering probability under $g$ is:
\[
\eta \bigl(H(g)\bigr)
  =\frac12-\sqrt{\frac{\log 2-H(g)}{2}},
  \tag{5}
\]
\end{definition}
\subsection{Primary Contribution}
\label{subsec:main}

We look to understand the contribution of these data selection factors to model optimisation in tandem. Thus, we establish a generalization bound for the special case of ranking distillation. We apply a PAC bound; these bounds affect a model's preference for a particular hypothesis within a hypothesis class, effectively governing how a model will generalise.

\begin{theorem}[Ranking Distillation Generalisation Bound (Proof in Appendix \ref{app:risk})]
\label{thm:main}
Let $(\mathcal X,d, \mu_Q)$ be a metric-measure space for $Q$ and let $\mathcal H$ be a hypothesis class of VC dimension $d$ such that every $h\in\mathcal H$ is $L$-Lipschitz. Let $f^{\star}=\arg\min_{h\in\mathcal H}\mathcal R(h)$ and let $\widehat f$ minimise an empirical semi-supervised Bregman loss with potential $\phi$ satisfying $|\phi'|\le1$. Then for every confidence level $\delta\in(0,1)$, with probability at least $1-\delta$,
\[
  \mathcal R(\widehat f)-\mathcal R(f^{\star})
  \;\le\;
  \zeta\,L\,\Delta_Q\,
        \eta \bigl(H(g)\bigr)
  \;+\;
  C\sqrt{\frac{d\log(1/\delta)}{n}},
  \tag{6}
\]
where $\zeta$ depends only on divergence potential~$\phi$, and $C>0$ is an absolute universal constant and $n$ is the number of observed instances.
\end{theorem}
\label{subsec:implications}

This bound indicates that teacher selection should achieve a sufficient level of entropy which should be tightened when locality is loose. Further, incorporating unlabeled data through a semi-supervised approach can improve performance by yielding more accurate estimates of the class diameter, and although the choice of Bregman loss affects the constant $\zeta$ in the bound, it does not alter the fundamental scaling behavior. The metric structure manifests through the essential diameter $\Delta_Q$, which captures how ``spread out'' the relevant items are around each query. This geometric constraint affects generalization in biased sampling scenarios such as negative mining. This means that a biased sampler, where pairs cannot be drawn under the true measure $\mu_Q$, should be addressed.

\subsection{Effect of Biased Sampling: Density-Ratio Adjustment}
\label{sec:bias}

In practice, we often use biased sampling strategies to select training examples. Let a deterministic sampler pick a subset $\mathcal{X}'_Q \subset \mathcal{X}_Q$ and write $\nu_Q$ for the induced measure. We can analyze how this biased sampling affects our bounds following \citet{hsu:2021}, establishing the density ratio $\kappa_Q,$ bridging to our original measure via:
\[
  \kappa_Q
  = \sup_{x\in\operatorname{supp}\nu_Q}
      \frac{d\mu_Q}{d\nu_Q}(x)
  < \infty .
  \tag{7}
\]

\begin{corollary}[Fixed-miner density–ratio bound (Proof in Appendix \ref{app:biased})] Under a biased sampling measure with density ratio $\kappa_Q$ with respect to $\mu_Q$, excess risk is expressed through Theorem 1 by reweighting for our original measure $\mu_Q$:
\[
  \mathcal R_{\mu_Q}(\widehat f)-\mathcal R_{\mu_Q}(f^{\star})
  \le
  \zeta\,L\Delta_Q\,\eta\bigl(H(g)\bigr)
  + C\sqrt{\frac{\kappa_Q\,d\,\log(1/\delta)}{n}} .
\tag{8}
\]
\end{corollary}

\noindent
When the pool is \emph{in-domain} ($\kappa_Q\approx1$) the term is
unchanged; off-domain pools inflate the bound.
Empirical values of $H_{\nu}(g)$, $\Delta_Q$, and $\kappa_Q$ are reported in Table \ref{tab:geometry}.

\section{Empirical Investigation}
Having established theoretical factors in ranking distillation, we now investigate common settings from literature.\footnote{Repository: \href{https://github.com/Parry-Parry/locality-and-entropy}{https://github.com/Parry-Parry/locality-and-entropy}}

\subsection{Experimental Setup}

\noindent \textbf{Datasets and Metrics} We assess in-domain effectiveness on the TREC Deep Learning 2019~\citep{craswell:2019} and 2020~\citep{craswell:2020} test collections, retrieving from the MSMARCO passage collection~\citep{nguyen:2016}. To assess out-of-domain effectiveness, we use all public test collections comprising the BEIR benchmark~\citep{thakur:2021}. We report dataset statistics and full out-of-domain results in Appendix \ref{app:ood}. All reported test collections apply normalized discounted cumulative gain (nDCG)~\citep{DBLP:journals/tois/JarvelinK02} as their primary measure and we report this value at a rank cutoff of 10 as is standard, to provide insight at greater recall values we also report mean averaged precision (MAP) over full rankings. We use a TOST (Two One-Sided T-test) to assess statistical equivalence with $\alpha=0.05$ and means bounded at $5\%$. Further details are provided in Appendix \ref{app:details}.

\noindent \textbf{Models} We train both cross-encoders and bi-encoders to assess how representation capacity affects investigated training settings. All cross-encoders are initialised from an ELECTRA checkpoint~\citep{clark:2020}, as this model has been shown empirically to yield greater effectiveness than a standard BERT model for cross-encoder initialisation~\citep{pradeep:2022}. We initialise all bi-encoders from the original BERT checkpoint~\citep{devlin:2019}. Each model is a standard transformer encoder with twelve layers and six attention heads per layer. 

To ensure reproducibility and clear attribution of effectiveness, we train a cross-encoder following \citep{pradeep:2022} using the ELECTRA architecture trained for one epoch with BM25 localised negatives on the MSMARCO-passage training set. Under distillation, it is generally unnecessary~\citep{althammer:2023, schlatt:2025} and computationally infeasible in our study to train all model variants at this scale; thus, this model acts as a teacher trained in a data-rich environment.

\noindent \textbf{Loss Criteria}
As a supervised criterion, we employ localized contrastive estimation (LCE)~\citep{gao:2021}, similar to infoNCE~\citep{gutmann:2010} and more generally model-conditioned NCE~\citep{ceylan:2018}. This criterion, assuming $\mathcal{X}_i$ is a known positive, is defined as:
{\small{\[
\ell_{\mathrm{LCE}}(\mathcal{X};\,f)
=
-\,\log
\frac{
  \exp \bigl(f(\mathcal{X}_i)/\tau\bigr)
}{
  \exp \bigl(\mathcal{X}_i/\tau\bigr)
  +
  \sum_{j=1, j \neq i}^{m-1} \exp \bigl(\mathcal{X}_j/\tau\bigr)
}.
\tag{9}
\]}}

We employ three semi-supervised criteria prevalent in neural ranking literature. The first marginMSE aims to reduce the effect of differences in ranking modes by optimising the margin between positive ($\mathcal{X}_i$) and negative ($\mathcal{X}_j$) elements instead of pointwise scores~\citep{hofstätter:2021}.
{\small{\[
\ell_{\mathrm{marginMSE}}(\mathcal{X};f,g)
\;=\;
\sum_{\substack{j=1\\j\neq i}}^{m}
\Bigl[
  f(\mathcal{X}_i) - f(\mathcal{X}_j)\bigr)
  -
  g(\mathcal{X}_i) + g(\mathcal{X}_j)\bigr)
\Bigr]^2
\tag{10}
\]}}

The second, RankNet, is common in learning-to-rank literature~\citep{burges:2010}. It sees increasing application with increasing precision of modern ranking models as it optimizes all $x, x'$ interactions agnostic of human labels.

{\small{\[
  \mathcal \ell_{\text{RankNet}}(f) = \sum_{i=1}^{|\mathcal{X}|} \sum_{\substack{j=1\\j\neq i}}^{|\mathcal{X}|} -\Bigl[y_{ij}\,\log \sigma(s_{ij}) + (1-y_{ij})\,\log\bigl(1-\sigma(s_{ij})\bigr)\Bigr], \tag{11}
\]}}

where \( s_{ij}=f(\mathcal{X}_i)-f(\mathcal{X}_j) \), \( \sigma(z)=\frac{1}{1+e^{-z}} \), and \( y_{ij} \in \{0,1\} \). 

Finally, KL divergence is commonly used as a loss criterion in several settings beyond ranking~\citep{lin:2020, kingma:2022}. It assumes, much like RankNet, that a reference distribution, in this case $[g(\mathcal{X}_i)]_1^m$, represents the absolute ground truth.

{\small{\[
\ell_{\mathrm{KL}}(\mathcal{X};f,g)
\;=\;
\sum_{j=1}^m
f(\mathcal{X})_j
\;\log\!\frac{f(\mathcal{X})_j}{g(\mathcal{X})_j}
\tag{12}
\]}}

We show these semi-supervised criteria can be expressed as Bregman divergences in Appendix \ref{app:bregman}.

\noindent \textbf{Implementation Details} All models observe a total of $12M$ documents from the MSMARCO-passage collection\citep{nguyen:2016}, providing approximately equal computational budget across all training settings; this is equivalent to observing $93k$ topics with $16$ documents each. We follow \citep{pradeep:2022} in setting the learning rate of cross-encoder runs to $1e{-5}$ and \citep{hofstätter:2021} for bi-encoders, setting the learning rate to $7e{-6}$. We apply a learning rate warmup for 0.1 epochs and then linear decay with an AdamW optimizer using default hyperparameter settings~\citep{loshchilov:2019}. We use the Transformers library~\citep{wolf:2020} with a PyTorch backend~\citep{paszke:2019} for all training processes. We use PyTerrier for inference and evaluation~\citep{pyterrier:2020}. 

\noindent \textbf{Negative Sampling Sources} \looseness=-1 We consider four sampling sources in our investigation, aligning with those applied in literature. The first is uniform selection from the training corpus (Random). The second is a lexical heuristic BM25 ($k_1=1.2$, $b=0.75$)~\citep{robertson:1995}, a lightweight retrieval model. The third is our teacher model, a cross-encoder (CE). Finally, we apply the ensemble pipeline shown in Figure \ref{fig:sampling} (Ensemble), we use the rankings supplied but filter by our teacher model to ensure fairness across settings and outline all models contained within this ensemble in Appendix \ref{app:details}. 

\subsection{Discussion}

\begin{table}[t]
  \centering
  \footnotesize
  \caption{\small Ranking effectiveness across loss criteria and sampling domains. In-domain effectiveness is evaluated on TREC Deep Learning test collections, and out-of-domain effectiveness is evaluated on the BEIR benchmark (per-dataset effectiveness is shown in Appendix \ref{app:ood}). Superscripts denote statistical equivalence via TOST (1\% bound, $\alpha=0.05$).}
  \setlength{\tabcolsep}{2pt}
    \begin{adjustbox}{width=0.8\columnwidth}
  
  \begin{tabular}{llcccccccccc}
  \toprule
     & & \multicolumn{4}{c}{TREC DL'19} & \multicolumn{4}{c}{TREC DL'20} & \multicolumn{2}{c}{BEIR} \\
  \cmidrule(lr){3-6}\cmidrule(lr){7-10}\cmidrule(lr){11-12}
  & & \multicolumn{2}{c}{BE} & \multicolumn{2}{c}{CE} & \multicolumn{2}{c}{BE} & \multicolumn{2}{c}{CE} & BE & CE \\
  \cmidrule(lr){3-4}
  \cmidrule(lr){5-6}
  \cmidrule(lr){7-8}
  \cmidrule(lr){9-10}
  \cmidrule(lr){11-11}
  \cmidrule(lr){12-12}
   Loss & Domain & nDCG & MAP & nDCG & MAP & nDCG & MAP & nDCG & MAP & nDCG & nDCG \\
  \midrule
  LCE & Random & 0.546\textsuperscript{\phantom{aaa}} & 0.333\textsuperscript{\phantom{aaa}} & 0.628\textsuperscript{bcd\phantom{}} & 0.407\textsuperscript{bcd\phantom{}} & 0.523\textsuperscript{\phantom{aaa}} & 0.360\textsuperscript{\phantom{aaa}} & 0.615\textsuperscript{\phantom{aaa}} & 0.406\textsuperscript{\phantom{aaa}} & 0.459\textsuperscript{bc\phantom{a}} & 0.507\textsuperscript{d\phantom{aa}} \\
  LCE & BM25 & 0.642\textsuperscript{cd\phantom{a}} & 0.385\textsuperscript{cd\phantom{a}} & 0.681\textsuperscript{acd\phantom{}} & 0.459\textsuperscript{acd\phantom{}} & 0.622\textsuperscript{cd\phantom{a}} & 0.425\textsuperscript{cd\phantom{a}} & 0.686\textsuperscript{cd\phantom{a}} & 0.488\textsuperscript{cd\phantom{a}} & 0.462\textsuperscript{ac\phantom{a}} & 0.472\textsuperscript{\phantom{aaa}} \\
  LCE & CE & 0.634\textsuperscript{bd\phantom{a}} & 0.389\textsuperscript{bd\phantom{a}} & 0.675\textsuperscript{abd\phantom{}} & 0.454\textsuperscript{abd\phantom{}} & 0.632\textsuperscript{bd\phantom{a}} & 0.424\textsuperscript{bd\phantom{a}} & 0.689\textsuperscript{bd\phantom{a}} & 0.482\textsuperscript{bd\phantom{a}} & 0.456\textsuperscript{ab\phantom{a}} & 0.478\textsuperscript{\phantom{aaa}} \\
  LCE & Ensemble & 0.658\textsuperscript{bc\phantom{a}} & 0.397\textsuperscript{bc\phantom{a}} & 0.730\textsuperscript{abc\phantom{}} & 0.501\textsuperscript{abc\phantom{}} & 0.655\textsuperscript{bc\phantom{a}} & 0.448\textsuperscript{bc\phantom{a}} & 0.738\textsuperscript{bc\phantom{a}} & 0.520\textsuperscript{bc\phantom{a}} & 0.459\textsuperscript{\phantom{aaa}} & 0.502\textsuperscript{a\phantom{aa}} \\
  \midrule
  RankNet & Random & 0.336\textsuperscript{\phantom{aaa}} & 0.199\textsuperscript{\phantom{aaa}} & 0.616\textsuperscript{\phantom{aaa}} & 0.409\textsuperscript{\phantom{aaa}} & 0.273\textsuperscript{\phantom{aaa}} & 0.185\textsuperscript{\phantom{aaa}} & 0.578\textsuperscript{\phantom{aaa}} & 0.404\textsuperscript{\phantom{aaa}} & 0.323\textsuperscript{\phantom{aaa}} & 0.445\textsuperscript{\phantom{aaa}} \\
  RankNet & BM25 & 0.653\textsuperscript{cd\phantom{a}} & 0.410\textsuperscript{cd\phantom{a}} & 0.731\textsuperscript{cd\phantom{a}} & 0.485\textsuperscript{cd\phantom{a}} & 0.659\textsuperscript{cd\phantom{a}} & 0.424\textsuperscript{cd\phantom{a}} & 0.739\textsuperscript{cd\phantom{a}} & 0.498\textsuperscript{cd\phantom{a}} & 0.468\textsuperscript{\phantom{aaa}} & 0.527\textsuperscript{c\phantom{aa}} \\
  RankNet & CE & 0.657\textsuperscript{bd\phantom{a}} & 0.409\textsuperscript{bd\phantom{a}} & 0.721\textsuperscript{bd\phantom{a}} & 0.489\textsuperscript{bd\phantom{a}} & 0.651\textsuperscript{bd\phantom{a}} & 0.432\textsuperscript{bd\phantom{a}} & 0.734\textsuperscript{bd\phantom{a}} & 0.496\textsuperscript{bd\phantom{a}} & 0.475\textsuperscript{\phantom{aaa}} & 0.526\textsuperscript{b\phantom{aa}} \\
  RankNet & Ensemble & 0.679\textsuperscript{bc\phantom{a}} & 0.413\textsuperscript{bc\phantom{a}} & 0.719\textsuperscript{bc\phantom{a}} & 0.483\textsuperscript{bc\phantom{a}} & 0.689\textsuperscript{bc\phantom{a}} & 0.459\textsuperscript{bc\phantom{a}} & 0.747\textsuperscript{bc\phantom{a}} & 0.508\textsuperscript{bc\phantom{a}} & 0.494\textsuperscript{\phantom{aaa}} & 0.488\textsuperscript{\phantom{aaa}} \\
  \midrule
  mMSE & Random & 0.602\textsuperscript{bcd\phantom{}} & 0.374\textsuperscript{bcd\phantom{}} & 0.693\textsuperscript{bcd\phantom{}} & 0.433\textsuperscript{bcd\phantom{}} & 0.637\textsuperscript{bcd\phantom{}} & 0.415\textsuperscript{bcd\phantom{}} & 0.685\textsuperscript{bcd\phantom{}} & 0.470\textsuperscript{bcd\phantom{}} & 0.459\textsuperscript{\phantom{aaa}} & 0.477\textsuperscript{\phantom{aaa}} \\
  mMSE & BM25 & 0.662\textsuperscript{acd\phantom{}} & 0.411\textsuperscript{acd\phantom{}} & 0.601\textsuperscript{acd\phantom{}} & 0.390\textsuperscript{acd\phantom{}} & 0.666\textsuperscript{acd\phantom{}} & 0.450\textsuperscript{acd\phantom{}} & 0.607\textsuperscript{acd\phantom{}} & 0.400\textsuperscript{acd\phantom{}} & 0.467\textsuperscript{\phantom{aaa}} & 0.520\textsuperscript{cd\phantom{a}} \\
  mMSE & CE & 0.676\textsuperscript{abd\phantom{}} & 0.422\textsuperscript{abd\phantom{}} & 0.724\textsuperscript{abd\phantom{}} & 0.484\textsuperscript{abd\phantom{}} & 0.668\textsuperscript{abd\phantom{}} & 0.448\textsuperscript{abd\phantom{}} & 0.737\textsuperscript{abd\phantom{}} & 0.511\textsuperscript{abd\phantom{}} & 0.473\textsuperscript{\phantom{aaa}} & 0.523\textsuperscript{bd\phantom{a}} \\
  mMSE & Ensemble & 0.683\textsuperscript{abc\phantom{}} & 0.414\textsuperscript{abc\phantom{}} & 0.717\textsuperscript{abc\phantom{}} & 0.482\textsuperscript{abc\phantom{}} & 0.661\textsuperscript{abc\phantom{}} & 0.458\textsuperscript{abc\phantom{}} & 0.736\textsuperscript{abc\phantom{}} & 0.516\textsuperscript{abc\phantom{}} & 0.492\textsuperscript{\phantom{aaa}} & 0.523\textsuperscript{bc\phantom{a}} \\
  \midrule
  KL & Random & 0.571\textsuperscript{\phantom{aaa}} & 0.361\textsuperscript{\phantom{aaa}} & 0.661\textsuperscript{bcd\phantom{}} & 0.428\textsuperscript{bcd\phantom{}} & 0.529\textsuperscript{\phantom{aaa}} & 0.356\textsuperscript{\phantom{aaa}} & 0.625\textsuperscript{\phantom{aaa}} & 0.416\textsuperscript{\phantom{aaa}} & 0.447\textsuperscript{\phantom{aaa}} & 0.511\textsuperscript{c\phantom{aa}} \\
  KL & BM25 & 0.655\textsuperscript{cd\phantom{a}} & 0.401\textsuperscript{cd\phantom{a}} & 0.698\textsuperscript{acd\phantom{}} & 0.471\textsuperscript{acd\phantom{}} & 0.637\textsuperscript{cd\phantom{a}} & 0.430\textsuperscript{cd\phantom{a}} & 0.726\textsuperscript{cd\phantom{a}} & 0.508\textsuperscript{cd\phantom{a}} & 0.466\textsuperscript{c\phantom{aa}} & 0.504\textsuperscript{d\phantom{aa}} \\
  KL & CE & 0.660\textsuperscript{bd\phantom{a}} & 0.401\textsuperscript{bd\phantom{a}} & 0.712\textsuperscript{abd\phantom{}} & 0.477\textsuperscript{abd\phantom{}} & 0.633\textsuperscript{bd\phantom{a}} & 0.434\textsuperscript{bd\phantom{a}} & 0.728\textsuperscript{bd\phantom{a}} & 0.509\textsuperscript{bd\phantom{a}} & 0.467\textsuperscript{b\phantom{aa}} & 0.513\textsuperscript{a\phantom{aa}} \\
  KL & Ensemble & 0.660\textsuperscript{bc\phantom{a}} & 0.402\textsuperscript{bc\phantom{a}} & 0.727\textsuperscript{abc\phantom{}} & 0.494\textsuperscript{abc\phantom{}} & 0.670\textsuperscript{bc\phantom{a}} & 0.455\textsuperscript{bc\phantom{a}} & 0.733\textsuperscript{bc\phantom{a}} & 0.519\textsuperscript{bc\phantom{a}} & 0.485\textsuperscript{\phantom{aaa}} & 0.463\textsuperscript{b\phantom{aa}} \\
  \bottomrule
  \end{tabular}
  \end{adjustbox}
  \label{tab:main}
\end{table}

\textbf{Effectiveness under Different Empirical Distributions}
Our ablation of sampling distributions under semi-supervision shows multiple cases where investing computational budget in a strong estimator is unnecessary to improve generalisation both in- and out-of-domain. In Table \ref{tab:main}, rows 1-4 show effectiveness in a supervised setting using a contrastive objective. In-domain, it is clear that localised sampling by heuristics can be effective as there is a clear trend in effectiveness as ``hardness'' and thus computation is expended. Out-of-domain, observe that generally, where a single estimator-induced distribution is insufficient to cover the query manifold, a random or ensemble sample yields greater robustness. We find that, though in explicit contrastive learning, there is a clear trend that the tighter a sampling distribution is, the more model effectiveness can continue to improve; we find that under all semi-supervision settings, effectiveness plateaus once a minimal locality is enforced (BM25) with inconsistent effectiveness improvements suggesting that heuristics such as mining from rankings are insufficient to explain effectiveness gains. Though bias induced by sampling is indeed reduced by ensembling approaches, empirical values of the query-specific diameter show that the query space does not become more compact. This is shown in Table~\ref{tab:geometry}, explaining minimal change in out-of-domain effectiveness as our bias term is largely unchanged across domains. However, aligning with corollary 3.1.1, we see that when density ratios are minimised across our settings under an ensembling approach, a Bi-Encoder continues to improve, suggesting this term can have a larger effect; nevertheless, we do not find settings where a statistically differentiable positive effect is found across multiple domains (e.g helping both in and out-of-domain) when applying computationally expense data selection strategies. 

\begin{table}[t]
  \centering
  \footnotesize
  \setlength{\tabcolsep}{2pt}
  \caption{\small Empirical Values of teacher entropy, the relative density ratio of each sampling domain and empirical measures of the query diameter. Each is taken at the $95^\text{th}$ percentile (either max or mean) for robust estimates. Entropy is measured in Nats via Shannon entropy over each ranking. We note that these are approximations of the theoretical aspects discussed in this work.}
  \renewcommand{\arraystretch}{1.2}
  \begin{tabular}{@{}lccc@{}}
  \toprule
  Source & $\widehat H_{\nu}(g)$ & $\widehat \kappa_Q$ & $\widehat\Delta_Q$\\
  \midrule
  Random & 6.62$\pm 0.127$ & 14.202$\pm 556.251$ & 10.448$\pm 0.044$ \\
  BM25 & 4.973$\pm 1.978$ & 12.747$\pm 461.588$ & 9.862$\pm 0.205$ \\
  Cross-Encoder & 4.068$\pm 0.930$ & 11.116$\pm 353.018$ & 9.593$\pm 0.251$ \\
  Ensemble & 3.973$\pm 0.838$ &  8.276$\pm 165.234$ & 9.546$\pm 0.233$ \\
  \bottomrule
  \end{tabular}
\label{tab:geometry}
\end{table}

\textbf{Effectiveness under Different Entropy}
Having controlled for entropy in different training settings, we now fix the sampling domain $\nu_Q$, choosing BM25 and select rankings based on where they lie within the sampled teacher entropy distribution by quartile. In Table \ref{tab:teacher-entropy}, we see that once locality is established, one can further improve effectiveness in-domain by selecting examples conditioned on ranking entropy. Even under a constrastive setting we find that sub-selection by some teacher (significantly less expensive then ensembling before similarly filtering) can further improve performance reducing the gap observed in our main results under a contrastive objective. We see that choosing the central mass of the entropy distribution (inner quartiles) is most effective and can be contrasted with selection in the outlier quartiles in which effectiveness degrades. Generally we observe that the upper quartile will yield greater effectiveness over the lower quartile, coupled with the reduced effectiveness of selection via outlier quartiles we consider that a balance must be struck between capturing high entropy examples for the purposes of generalisation. We note that all cases degrade out-of-domain potentially suggesting that choosing examples by these criteria induce overfitting to the particular cases within the training domain. We observe that increased mean entropy leads to reduced effectiveness, considering Table \ref{tab:geometry}, we can infer that in cases where entropy is high, the representation space is insufficiently tight to compensate for this entropy. Out-of-domain correlation is minimal as the density ratio $\kappa_Q$ will inflate the VC term, bounding generalisation under this setting, thus any attempt to improve locality will fail to improve effectiveness.

\begin{table}[t]
  \centering
  \footnotesize
  \setlength{\tabcolsep}{3pt}
  \caption{\small IR effectiveness across domain subsets by quartiles (Q) of the teacher entropy distribution over training examples. Effectiveness is measured and evaluated as noted in Table \ref{tab:main}. Shannon Entropy is denoted $\widehat H_\nu(g)$.}
\begin{adjustbox}{width=.6\textwidth}

  \begin{tabular}{@{}lllccccc@{}}
  \toprule
  & & & \multicolumn{2}{c}{TREC DL’19} & \multicolumn{2}{c}{TREC DL’20} & \multicolumn{1}{c}{BEIR} \\
  \cmidrule(lr){4-5}
  \cmidrule(lr){6-7}
  \cmidrule(lr){8-8}
    Loss & Transform & $\widehat H_\nu(g)$ & nDCG & MAP & nDCG & MAP & nDCG \\
  \midrule
  LCE & Original & 4.973$\pm 1.978$ & 0.681 & 0.459 & 0.686 & 0.488 & 0.472 \\
  LCE & Lower Q & 5.814$\pm1.483$ & 0.690 & 0.469 & 0.720 & 0.506 & 0.454 \\
  LCE & Inner Qs & 5.344$\pm 1.527$ & 0.723 & 0.491 & 0.742 & 0.520 & 0.465 \\
  LCE & Upper Q & 5.106$\pm 1.027$ & 0.716 & 0.482 & 0.739 & 0.518 & 0.460 \\
  LCE & Outlier Qs & 6.807$\pm 1.278$ & 0.638 & 0.412 & 0.636 & 0.425 & 0.357 \\
  \midrule
  mMSE & Original & 4.973$\pm 1.978$ & 0.601 & 0.390 & 0.607 & 0.400 & 0.520 \\
  mMSE & Lower Q & 5.814$\pm1.483$ & 0.720 & 0.485 & 0.729 & 0.511 & 0.484 \\
  mMSE & Inner Qs & 5.344$\pm 1.527$ & 0.727 & 0.492 & 0.729 & 0.505 & 0.490 \\
  mMSE & Upper Q & 5.106$\pm 1.027$ & 0.724 & 0.491 & 0.737 & 0.504 & 0.491 \\
  mMSE & Outlier Qs & 6.807$\pm 1.278$ & 0.712 & 0.473 & 0.732 & 0.501 & 0.469 \\
  \bottomrule
\end{tabular}
\label{tab:teacher-entropy}
\end{adjustbox}
\end{table}

\textbf{Differences in Model Behaviour}
Even under densely annotated ($>6$ relevant texts per query~\citep{craswell:2019, craswell:2020}) test collections, variations in in-domain effectiveness between operational settings are minimal; this is expected via our bound. However, when observing intrinsic model behaviour, we observe that the chosen domain can greatly affect score distributions under otherwise identical optimization settings. In Figure \ref{fig:fourplots}, observe how different settings align with a power law; this can be considered alignment with a Zipfian distribution. See how optimization criteria lead to a vastly different score distribution from the teacher in Figure \ref{fig:4a}, the original teacher has high confidence within the top-10 ranks, and scores reach an elbow point at rank 218. However, depending on the sampling domain, we observe collapse when applying a random distribution in Figure \ref{fig:4b} as both entropy and the relative density ratio, as outlined in Table \ref{tab:geometry}, are insufficiently tight, leading to collapse. Conversely, we observe power law behaviour when applying an ensemble with an elbow at rank 12. Given the weak performance of ensemble approaches out-of-domain when this behaviour is present, we posit that this highly confident behaviour may be overfitting to an in-domain setting and may not be desirable. Though we leave any causal analysis to future work, we observe this behavior in several settings as shown in Appendix \ref{app:figs}.

\begin{figure}[htb]
  \centering
  \begin{subfigure}[b]{0.32\columnwidth}
    \centering
    \includegraphics[width=\linewidth]{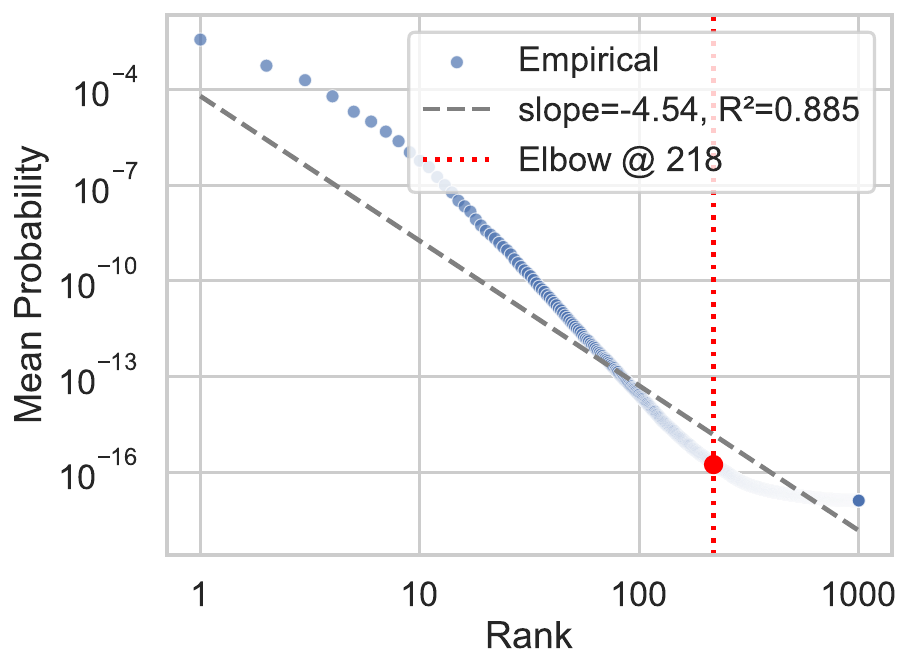}
    \caption{Teacher Model}\label{fig:4a}
  \end{subfigure}\hfill
  \begin{subfigure}[b]{0.32\columnwidth}
    \centering
    \includegraphics[width=\linewidth]{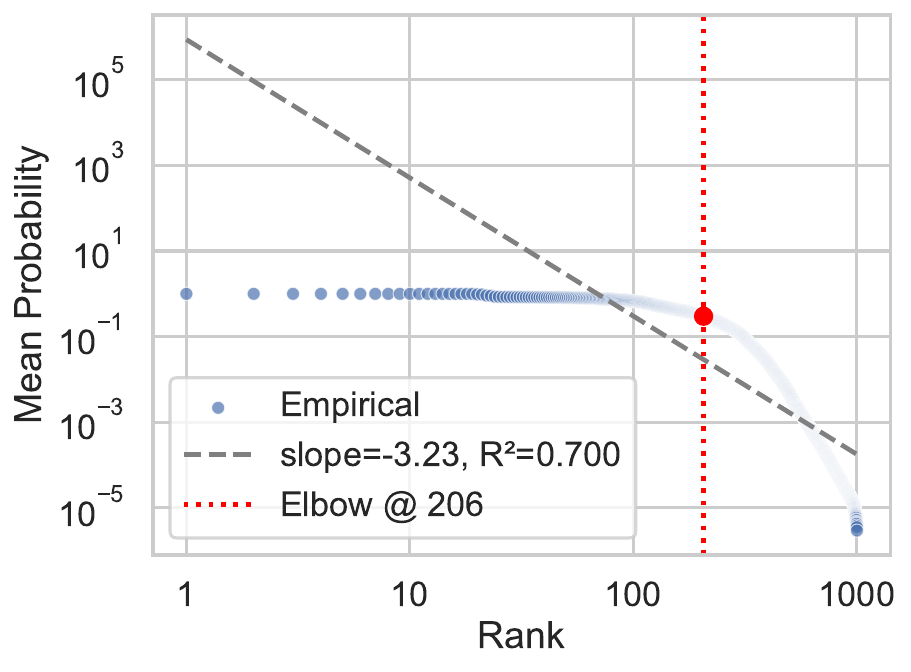}
    \caption{RankNet - Random}\label{fig:4b}
  \end{subfigure}\hfill
  \begin{subfigure}[b]{0.32\columnwidth}
    \centering
    \includegraphics[width=\linewidth]{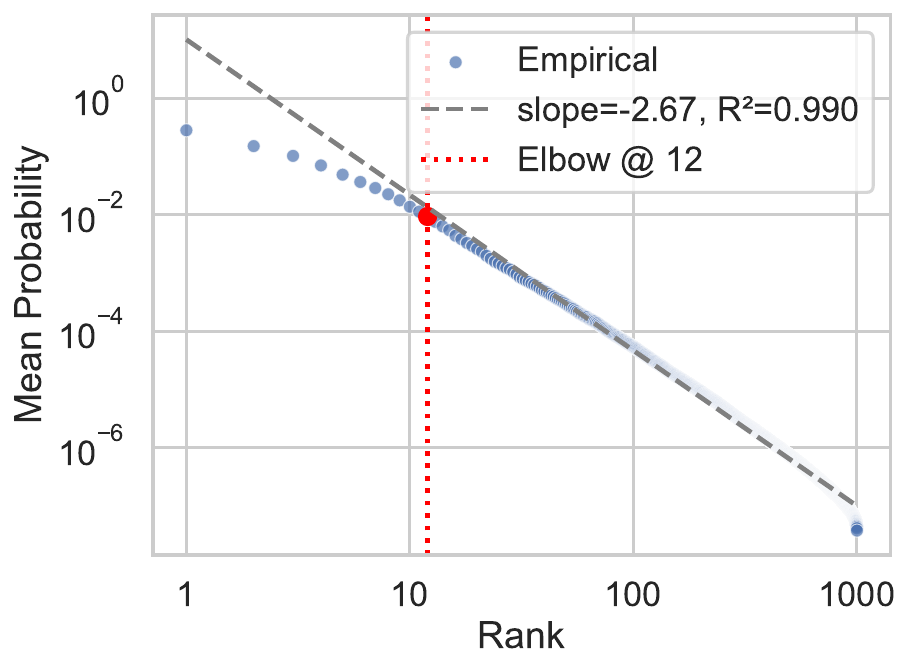}
    \caption{RankNet - Ensemble}\label{fig:4c}
  \end{subfigure}\hfill
  \caption{Score versus rank ratio comparing the teacher (LCE) and two RankNet-trained students with different sampling distributions evaluated on TREC DL`19. Note the log-log scale and alignment with power laws plotted in grey.}
  \label{fig:fourplots}
\end{figure}

\section{Conclusion}
In this work, we have provided a systematic analysis of two core components in modern neural ranking pipelines—example locality (negative sampling) and target entropy (distillation)—both theoretically and empirically. Our analysis establishes a novel generalisation bound for ranking distillation in terms of locality and entropy, accounting for biased sampling strategies. This bound shows that overly “hard” or overly uniform teacher distributions can both degrade student performance, and that geometry-driven sampling impacts only the bias term, not the entropy term. Empirically, across both in-domain (TREC Deep Learning 2019/2020) and out-of-domain (BEIR) benchmarks, we demonstrate that complex, multi-stage hard-negative pipelines yield minimal gains over simpler sampling strategies under distillation and theoretically justify cases where it is valuable to expend such computation. Furthermore, by stratifying examples according to teacher ranking entropy, we observe consistent in-domain improvements at intermediate entropy levels, while high-entropy “outlier” subsets degrade performance, confirming the bound’s prediction. Moving forward, our findings encourage a shift away from computationally intensive ensemble and cascade architectures toward more principled training protocols that directly control for geometry (locality) and information content (entropy). By focusing on the two orthogonal dimensions we have identified, tangible improvements can be made in ranking effectiveness without excessive computational expense.

\bibliography{bibliography}

\appendix
\appendix

\section{Preliminaries and Definitions}
\label{app:bregman}

For completeness, we collect all loss functions used in the main text and show that each can be expressed as a \emph{Bregman divergence}~\citep{bregman:1967}. We recall the definition first.

\begin{definition}[Bregman divergence]
Let $\phi:\mathbb R\to\mathbb R$ be strictly convex and $C^{1}$. The Bregman divergence between $a,b\in\mathbb R$ is
\[
  D_{\phi}(a\Vert b)=\phi(a)-\phi(b)-\phi'(b)(a-b).
\]
\end{definition}

The Kullback–Leibler divergence can be formed as a Bregman divergence as shown by \citet{banerjee:2005}.

\subsection{RankNet}
RankNet \citep{burges:2010} minimises the pair-wise logistic loss
\[
  \mathcal L_{\text{RankNet}}(f) = -\Bigl[y_{ij}\,\log \sigma(s_{ij}) + (1-y_{ij})\,\log\bigl(1-\sigma(s_{ij})\bigr)\Bigr],
\]
where \( s_{ij}=f(\mathcal{X}_i)-f(\mathcal{X}_j) \), \( \sigma(z)=\frac{1}{1+e^{-z}} \), and \( y_{ij} \in \{0,1\} \). 

When \( y_{ij} \in \{0,1\} \), the Bregman divergence simplifies to:
\[
  D_{\phi}\bigl(y_{ij}\Vert p_{ij}(f)\bigr) = \mathcal L_{\text{RankNet}}(f),
\]
Note that while \( \phi(u) \) is undefined at \( u = 0 \) and \( u = 1 \), the Bregman divergence \( D_\phi(y_{ij} \Vert p_{ij}) \) is well-defined for \( y_{ij} \in \{0,1\} \) and \( p_{ij} \in (0,1) \) by taking limits, and equals the RankNet loss.
\subsection{MarginMSE}
MarginMSE \citep{hofstätter:2021} minimises the error between teacher and student margins
\[
  m_{ij}(f)=f(\mathcal{X}_i)-f(\mathcal{X}_j), \qquad m_{ij}(g)=g(\mathcal{X}_i)-g(\mathcal{X}_j).
\]
With the quadratic potential \( \phi(u)=u^{2} \) one obtains \( D_{\phi}(a\Vert b)=(a-b)^{2} \) and hence
\[
  \mathcal L_{\text{MarginMSE}}(f)=D_{\phi}\bigl(m_{ij}(f)\Vert m_{ij}(g)\bigr).
\]

\subsection{Definitions}

\begin{definition}[Polish metric space]
A \emph{Polish metric space} is a complete, separable metric space.
\end{definition}

\begin{definition}[Lipschitz continuity~\citep{weaver:1999}]
A function $h: (\mathcal{X}, d) \to \mathbb{R}$ is \emph{$L$-Lipschitz} if there exists $L > 0$ such that
$$|h(x) - h(x')| \leq L \cdot d(x, x')$$
for all $x, x' \in \mathcal{X}$.
\end{definition}

\begin{definition}[Pair-wise Risk]
For a scorer $f$ and query measure $\mu_Q$, the \emph{pair-wise risk} is the probability of mis-ordering:
\[\mathcal{R}(f) = \Pr_{(x,x') \sim \mu_Q^{\otimes 2}}[f(x) < f(x')].
\tag{2}
\]
\end{definition}

\begin{definition}[VC dimension for function classes]
For a hypothesis class $\mathcal{H} \subseteq \mathbb{R}^{\mathcal{X}}$, we consider the induced thresholded classes $\{x \mapsto \mathbf{1}[h(x) > t] : h \in \mathcal{H}, t \in \mathbb{R}\}$. The \emph{VC dimension} of $\mathcal{H}$ is the VC dimension of this induced class of indicator functions.
\label{def:vc-dim}
\end{definition}
\begin{lemma}[Hoeffding's Inequality]\label{lem:hoeffding}
Let $Z_{1},\dots,Z_{n}$ be independent real-valued r.v.’s with
$a_i\le Z_i\le b_i$ a.s.\ and $\mathbb E[Z_i]=\mu_i$. 
For every $t>0$ by Hoeffding's inequality~\citep{hoeffding:1963},
\[
  \Pr\ \Bigl[
    \bigl|\tfrac1n\sum_{i=1}^{n} Z_i - \tfrac1n\sum_{i=1}^{n} \mu_i\bigr|
    > t
  \Bigr]
  \;\le\;
  2\exp\ \Bigl(
    -\frac{2n^{2}t^{2}}{\sum_{i=1}^{n}(b_i-a_i)^{2}}
  \Bigr).
\]
In particular, if each $Z_i\ \in\ [0,1]$ then  
$\Pr\bigl[\,|\widehat{\mathcal R}_{n}-\mathcal R_{\mu_Q}|>t\bigr]
       \le 2e^{-2nt^{2}}$.
\end{lemma}

\begin{lemma}[VC uniform deviation]\label{lem:VC}
Let $\mathcal{H}\subseteq\mathbb R^{\mathcal X}$ be a hypothesis class
with $\operatorname{VCdim}(\mathcal{H})=d<\infty$ when thresholded
as in Definition \ref{def:vc-dim}.
For i.i.d.\ pairs $(x_s,x'_s)\sim\mu_Q^{\otimes2}$ set
\(Z_s=\ell_f(x_s,x'_s)\in[0,1]\).
Then for any $\delta\in(0,1)$, with probability at least $1-\delta$,
\[
  \sup_{h\in\mathcal H}
  \bigl|
     \widehat{\mathcal R}_{n}(h)-\mathcal R_{\mu_Q}(h)
  \bigr|
  \;\le\;
  C\,\sqrt{\frac{d\,\log(1/\delta)}{n}} ,
\]
where \(C>0\) is an absolute constant.
\end{lemma}

\begin{proof}[Sketch of Lemma \ref{lem:VC}]
For a fixed \(f\), apply Lemma \ref{lem:hoeffding} with $t=\epsilon$.
Sauer’s lemma bounds the growth function by
\(|\Pi_{\mathcal H}(n)|\le(en/d)^{d}\).
A union bound over all labelings gives  
\(\Pr[\sup_{f}|\,\widehat{\mathcal R}_{n}(f)-\mathcal R_{\mu_Q}(f)|>\epsilon]
      \le (en/d)^{d} 2e^{-2n\epsilon^2}\).
Solving for \(\epsilon\) and absorbing the logarithmic factor into a
universal \(C\) yields the stated bound.
\end{proof}

\begin{definition}[Essential diameter]
For a probability space $(\mathcal{X}, \mu_Q)$ equipped with metric $d$, the \emph{essential diameter} is
$$\Delta_Q = \text{ess sup}_{(x,x') \sim \mu_Q^{\otimes 2}} d(x, x') = \inf\{M \geq 0 : \mu_Q^{\otimes 2}(\{(x,x') : d(x,x') > M\}) = 0\}.$$
\end{definition}

\begin{definition}[Pair-wise entropy]
For a scorer $g: \mathcal{X} \to \mathbb{R}$ and measure $\mu_Q$ over $\mathcal{X}$, define
$$H(g) = -\mathbb{E}_{(x,x') \sim \mu_Q^{\otimes 2}}[p_{x,x'} \log p_{x,x'} + (1-p_{x,x'}) \log(1-p_{x,x'})],$$
where $p_{x,x'} = \mathbf{1}[g(x) > g(x')]$.
\end{definition}

\section{Bounding Misordering under Teacher Targets}
\label{app:misorder}

\label{app:pinsker-entropy}
\begin{lemma}[Lower Bound on Pair-wise Error]
\label{lem:pinsker}
For a fixed document pair $(x,x')$ under meaure $\mu_Q$, let
\[
  Z = \mathbf 1[g(x)>g(x')], 
  \quad
  p = \Pr[Z=1],
  \quad
  P = (p,1-p),
  \quad
  U = \Bigl(\tfrac12,\tfrac12\Bigr).
\]
Then the teacher's misordering probability (w.r.t. true pairwise ordering) satisfies
\[
  \epsilon(p) := \Pr[g(x) < g(x')]
  \ge \eta(H(p))
  := \frac12 - \sqrt{\frac{\log 2 - H(p)}{2}},
\]
where $H(p) = -p\log p - (1-p)\log(1-p)$ is the binary Shannon entropy (in nats), and $\eta$ is defined piecewise (monotonically decreasing) as:
\[
  \eta(h) = \frac12 - \sqrt{\frac{\log 2 - h}{2}}.
\]
Furthermore, if \( H(g) \leq H_0 \) for some \( H_0 < \log 2 \), then the expected misordering probability satisfies:
\[
  \mathbb{E}_{(x,x') \sim \mu_Q^{\otimes 2}} \Pr[g \text{ misorders } (x,x')] \geq \eta(H_0).
\]
\end{lemma}

\begin{proof}
By Pinsker's inequality following \citet{painsky:2020}, for distributions $P$ and uniform distribution $U$ :
\[
  \|P - U\|_{\mathrm{TV}} = | p - \tfrac12 | \leq \sqrt{\frac12 D_{\mathrm{KL}}(P \| U)}.
\]
We compute:
\begin{align*}
D_{\mathrm{KL}}(P || U) &= p \log \frac{p}{1/2} + (1-p)\log \frac{1-p}{1/2} \\
&= p \log(2p) + (1-p) \log(2(1-p)) \\
&= p \log p + p \log 2 + (1-p) \log (1-p) + (1-p) \log 2 \\
&= p \log p + (1-p) \log (1-p) + \log 2 \\
&= -H(p) + \log 2.
\end{align*}
Thus:
\[
  | p - \tfrac12 | \leq \sqrt{\frac{\log 2 - H(p)}{2}}.
\]
The misordering probability is $\epsilon(p) = \min(p,1-p) = \tfrac12 - | p - \tfrac12 |$. Substituting:
\[
  \epsilon(p) \geq \tfrac12 - \sqrt{\frac{\log 2 - H(p)}{2}} = \eta(H(p)).
\]

Now suppose $H(g) \leq H_0 < \log 2$. We have:
\[
  H(g) = \mathbb{E}_{(x,x') \sim \mu_q^{\otimes 2}} H(p_{x,x'}) \leq H_0,
\]
where $H(p_{x,x'})$ is the entropy of the Bernoulli trial $\mathbf{1}[g(x) > g(x')]$.

By Jensen's inequality, $H$ is concave. Using the convexity of $\eta(H(p)) = \tfrac12 - \sqrt{\frac{\log 2 - H(p)}{2}}$ (since $\frac{d^2}{dH^2}\ \big( \eta(H) \big) > 0$):
\[
  \mathbb{E}_{(x,x') \sim \mu_q^{\otimes 2}} \eta(H(p_{x,x'})) \geq \eta(H(g)).
\]
However, note that:
\[
  \mathbb{E}_{(x,x') \sim \mu_q^{\otimes 2}} \Pr[g \text{ misorders } (x,x')] = \mathbb{E}_{(x,x')} \epsilon(p_{x,x'}) \geq \mathbb{E}_{(x,x')} \eta(H(p_{x,x'})),
\]
hence we have:
\[
  \mathbb{E}[\text{misordering}]
  \geq \mathbb{E}[\eta(H(p_{x,x'}))]
  \geq \eta(H(g)).
\]
\end{proof}

\begin{corollary}[Query-level bound]
\label{cor:query-entropy}
Averaging Lemma~\ref{lem:pinsker} over $(x,x')\sim\mu_Q^{\otimes2}$ replaces $H(p)$ by the pair-wise entropy $H(g)$ defined in Section \ref{app:misorder}, giving
\[
  \mathbb{E}_{(x,x') \sim \mu_Q^{\otimes 2}} \Pr\ \bigl[g \text{ mis-orders }(x,x')\bigr]
  \;\ge\;
  \eta\ \bigl(H(g)\bigr).
\]
\end{corollary}
\section{Generalisation under risk from teacher entropy and locality}
\label{app:risk}
\begin{theorem}[Locality–Entropy excess-risk, unbiased case]
\label{thm:local-entropy-unbiased}
Assume the sample pairs $(\mathcal{X}_s,x'_s)_{s=1}^n$ are drawn
\emph{i.i.d.} from the same distribution $\mu_Q^{\otimes2}$
that defines the true risk.
Under the hypotheses stated in Section 2,
for any $\delta\ \in\ (0,1)$, with probability $\ge 1-\delta$
\[
  \boxed{%
    \mathcal R(\widehat f)-\mathcal R(f^\star)
    \;\le\;
    \zeta L\Delta_Q\,
    \eta\ \bigl(H(g)\bigr)
    +\;
    C\sqrt{\frac{d\,\log(1/\delta)}{n}}
  }.
\]
\end{theorem}

\begin{proof}

\noindent \textbf{Step 1. Teacher-Bayes Error.}
First, recall that \( f^\star \) is \( L \)-Lipschitz. Therefore, for any pair \((x, x')\),
\[
  | f^\star(x) - f^\star(x') | \leq L d(x, x') \leq L \Delta_Q.
\]
The entropy \( H(g) \) of the teacher model \( g \) is assumed to be sufficiently high to approximate \( f^\star \). Specifically, the difference in risks between \( g \) and \( f^\star \) can be bounded as:
\[
  \mathcal{R}(g) - \mathcal{R}(f^\star) = \mathbb{E}\bigl[ \mathbf{1}[g(x) < g(x')] - \mathbf{1}[f^\star(x) < f^\star(x')] \bigr].
\]
This can be further controlled by an increasing function \( \eta(H(g)) \) of the entropy:
\[
  \mathcal{R}(g) - \mathcal{R}(f^\star) \leq L \Delta_Q \eta(H(g)).
\]
Here, \( \eta(\cdot) \) is introduced to capture the fact that higher entropy \( H(g) \) might lead to larger deviations in ranking. The function \( \eta \) is constructed based on Pinsker’s inequality derived in Appendix \ref{app:misorder}.

\noindent \textbf{Step 2. Student-Teacher Calibration.} 
Next, we consider the difference between the student \( \hat{f} \) and the teacher \( g \). To bound this, we use the Bregman divergence \( Z_h \) between \( \delta_h \) and \( \delta_g \). By definition, \( Z_h \) measures the discrepancy in ranking between the student and teacher.
We derive a bound that relates the Bregman divergence to the 0-1 risk difference~\citep{painsky:2020}:
\[
  \mathcal{R}(\hat{f}) - \mathcal{R}(g) \leq \zeta \mathbb{E}[Z_{\hat{f}}].
\]
Here, \( \zeta \) is used to calibrate the Bregman divergence to the 0-1 loss.

\noindent \textbf{Step 3. Population Gap.}
Under the assumption that the Bregman divergence \( D_{\phi} \) is Lipschitz in its arguments, we have:  
\[
  Z_h(x,x') - Z_{f^\star}(x,x') \leq L | \delta_h(x,x') - \delta_{f^\star}(x,x') | \leq 2 L d(x,x'),
\]
where L is the Lipschitz constant of \( D_{\phi} \). Thus, taking expectations, we get bounds like:
\[
  \mathbb{E}[Z_{\hat{f}} | - \mathbb{E}[Z_{f^\star}] \leq 2L \Delta_Q.
\]
Alternatively, since \( h \mapsto Z_h(x,x') \) is L-Lipschitz under the score difference.
\[
  Z_{h}(x, x') - Z_{f^\star}(x, x') \leq | ( h(x) - h(x') ) - ( f^\star(x) - f^\star(x') ) | \leq 2L | d(x, x') | \leq 2L \Delta_Q.
\]
This gives the population-level bound:
\[
  \mathbb{E}[Z_{\hat{f}}] - \mathbb{E}[Z_{f^\star}] \leq L \Delta_Q.
\]

\noindent\textbf{Step 4. Uniform deviation (VC).}
For each fixed $h\in\mathcal H$ let
\[
  S_n(h)
  :=\widehat{\mathbb E}[Z_h]-\mathbb E[Z_h]
  =\frac1n\sum_{s=1}^{n}\bigl(Z_h(x_s,x'_s)-\mathbb E Z_h\bigr).
\tag{C.1}\label{eq:Sn}
\]
Because $0\le Z_h\le L\Delta_Q$ (Step 3) and the pairs are i.i.d.,
Hoeffding’s inequality (Lemma \ref{lem:hoeffding}) gives for any
$\varepsilon>0$
\[
  \Pr\bigl[\;|S_n(h)|>\varepsilon\bigr]\le
  2\exp\!\bigl(-\tfrac{2n\varepsilon^{2}}{(L\Delta_Q)^{2}}\bigr).
\tag{C.2}\label{eq:hoeffding-fixed}
\]

For a sample of size $n$ the family
$\{(x,x')\mapsto\mathbf 1[Z_h(x,x')>\tau]:h\in\mathcal H,\tau\in\mathbb R\}$
shatters at most
\[
  |\Pi_{\mathcal H}(n)|\;\le\;\Bigl(\frac{en}{d}\Bigr)^{d}
\tag{C.3}\label{eq:sauer}
\]
sign patterns (Sauer’s lemma~\citep{sauer:1972}).  Hence, the union bound over
$\mathcal H$ yields
\[
  \Pr\!\bigl[\sup_{h\in\mathcal H}|S_n(h)|>\varepsilon\bigr]
  \;\le\;
  2\Bigl(\frac{en}{d}\Bigr)^{d}
    \exp\!\Bigl(-\frac{2n\varepsilon^{2}}{(L\Delta_Q)^{2}}\Bigr).
\tag{C.4}\label{eq:union}
\]

\smallskip
Solve \eqref{eq:union} for the smallest $\varepsilon$ such that the
right-hand side is $\le\delta$:
\[
  \varepsilon_\delta
  \;=\;
  L\Delta_Q
  \sqrt{\frac{d\ln\!\bigl(2en/d\bigr)+\ln\!\bigl(2/\delta\bigr)}{2n}}
  \;\;=:\;C_L\sqrt{\frac{d\ln(1/\delta)}{n}},
\tag{C.5}\label{eq:eps-delta}
\]
where $C_L:=L\Delta_Q\sqrt{\tfrac12\bigl(1+\tfrac{\ln(2en/d)}{\ln(1/\delta)}\bigr)}$.
Thus, with probability $\ge1-\delta,$
\[
  \boxed{\;
    \sup_{h\in\mathcal H}|S_n(h)|\le C_L\sqrt{\frac{d\ln(1/\delta)}{n}}
  \;}
\tag{C.6}\label{eq:vc-deviation}
\]

\noindent\textbf{Step 5. Empirical optimality of $\widehat f$.}
Write
\[
\begin{aligned}
  \mathbb E[Z_{\widehat f}]
  -\mathbb E[Z_{f^\star}]
  &=
  \underbrace{%
      \bigl(\mathbb E[Z_{\widehat f}]
           -\widehat{\mathbb E}[Z_{\widehat f}]\bigr)}_{\text{(a)}}
  + \underbrace{%
      \bigl(\widehat{\mathbb E}[Z_{\widehat f}]
           -\widehat{\mathbb E}[Z_{f^\star}]\bigr)}_{\le0\;\text{(ERM)}}
  + \underbrace{%
      \bigl(\widehat{\mathbb E}[Z_{f^\star}]
           -\mathbb E[Z_{f^\star}]\bigr)}_{\text{(b)}}.
\end{aligned}
\tag{C.7}\label{eq:triangle}
\]
Both (a) and (b) are bounded in magnitude by the
uniform deviation \eqref{eq:vc-deviation}; the middle term is
non-positive.  Hence
\[
  \boxed{\;
    \mathbb E[Z_{\widehat f}]-\mathbb E[Z_{f^\star}]
    \;\le\;
    2\,C_L\sqrt{\frac{d\ln(1/\delta)}{n}}
  \;}
\tag{C.8}\label{eq:pop-gap}
\]

\noindent\textbf{Step 6. Final assembly.}
Add and subtract $\mathcal R(g)$, then plug in the
calibration (Step 2), the misordering term (Step 1),
and the gap \eqref{eq:pop-gap}:

\[
\begin{aligned}
  \mathcal R(\widehat f)-\mathcal R(f^\star)
  &= \bigl[\mathcal R(\widehat f)-\mathcal R(g)\bigr]
     +\bigl[\mathcal R(g)-\mathcal R(f^\star)\bigr] \\[4pt]
  &\le \zeta\,\mathbb E[Z_{\widehat f}]
       + L\Delta_Q\,\eta\!\bigl(H(g)\bigr) \\[4pt]
  &\le \zeta\Bigl(\mathbb E[Z_{f^\star}]
                  +2\,C_L\sqrt{\tfrac{d\ln(1/\delta)}{n}}\Bigr)
       + L\Delta_Q\,\eta\!\bigl(H(g)\bigr).
\end{aligned}
\tag{C.9}\label{eq:assemble-1}
\]

Step 3 supplied $\mathbb E[Z_{f^\star}]\le L\Delta_Q$, so
\[
  \mathcal R(\widehat f)-\mathcal R(f^\star)
  \;\le\;
  \underbrace{\zeta L\Delta_Q}_{\text{locality}}
  \;+\;
  \underbrace{L\Delta_Q\,\eta\!\bigl(H(g)\bigr)}_{\text{entropy}}
  \;+\;
  \underbrace{2\zeta C_L\sqrt{\tfrac{d\ln(1/\delta)}{n}}}_{\text{stat.}}.
\tag{C.10}\label{eq:final-three}
\]

If the teacher is not degenerate, i.e.\ $\eta(H(g))\ge\varepsilon>0$,
then $\zeta L\Delta_Q\le\zeta L\Delta_Q\,\varepsilon^{-1}\eta(H(g))$ and
the first two bias terms merge; absorbing all numerical
constants into a single $C$ gives the form we provide in the main body of this work, this reduced form is realistic given the nature of ranking model teachers in which degenerate cases are rare due to the smooth estimations provided by neural models through activation functions such as softmax. 

\end{proof}
\section{Biased-Sampling under Negative Miners}
\label{app:biased}
\begin{corollary}[Biased–sampling bound]
\label{cor:biased}
Let $\nu_Q$ be a retrieval distribution with
$\operatorname{sup}\nu_Q\subseteq\operatorname{sup}\mu_Q$ and define
\[
  \kappa_Q
  = \sup_{x\in\operatorname{supp}\nu_Q}
      \frac{d\mu_Q}{d\nu_Q}(x)
  < \infty .
\tag{E.1}\label{eq:kappa}
\]
Assume the hypotheses of Theorem~\ref{thm:local-entropy-unbiased}.
If the empirical minimiser $\widehat f$ is trained on
$n$ i.i.d.\ pairs
$(\mathcal{X}_s,\mathcal{X}'_s)\sim\nu_Q^{\otimes2}$, then for every $\delta\in(0,1)$,
with probability at least $1-\delta$,
\[
  \mathcal R_{\mu_Q}(\widehat f)-\mathcal R_{\mu_Q}(f^{\star})
  \le
  \zeta\,L\Delta_Q\,\eta\bigl(H(g)\bigr)
  + C\sqrt{\frac{\kappa_Q\,d\,\log(1/\delta)}{n}} .
\tag{E.2}\label{eq:bias-bound}
\]
\end{corollary}

\begin{proof}
\noindent \textbf{Step 1: Importance weights over sample measures.}
Let $\nu_Q$ be our biased measure, following \citet{hsu:2021}, we bridge to $\mu_Q$ via
\[
  w_Q(x)=\frac{d\mu_Q}{d\nu_Q}(x)\le\kappa_Q ,
\]
the true pair-wise risk is
\[
  \mathcal R_{\mu_Q}(h)
  =\mathbb E_{\nu_Q}[w_Q(x)\,\ell_h(x,x')],
  \qquad
  \ell_h=\mathbf 1\ \bigl[h(x)<h(x')\bigr].
\]

\noindent \textbf{Step 2: The Bregman loss class.}
Let
\[
  Z_h(x,x')
  = D_\phi\ \bigl(\delta_{x,x'}(h)\,\Vert\,\delta_{x,x'}(g)\bigr) ,
  \quad
  \delta_{x,x'}(h)=h(x)-h(x').
\]
Because $|Z_h|\le L\Delta_Q$ and $0<w_Q\le\kappa_Q$,
\[
  |w_Q Z_h|\le \kappa_Q L\Delta_Q .
\]

\noindent \textbf{Step 3: Uniform deviation under $\nu_Q$.}
Normalise by $L\Delta_Q$ and apply the VC symmetrisation bound;
there exists $C>0$ such that, with probability $1-\delta$,
\[
  \sup_{h\in\mathcal H}
  \Bigl|
    \mathbb E_{\nu_Q}[w_Q Z_h]
    - \frac1n\sum_{s=1}^{n} w_Q(\mathcal{X}_s) Z_h(\mathcal{X}_s,x'_s)
  \Bigr|
  \le
  C L\Delta_Q \sqrt{\frac{\kappa_Q\,d\,\log(1/\delta)}{n}} .
\tag{E.3}\label{eq:uniform}
\]

\noindent \textbf{Step 4: Use empirical optimality of $\widehat f$.}
By definition of $\widehat f$,
\[
  \frac1n\sum_{s=1}^{n} w_Q(\mathcal{X}_s) Z_{\widehat f}(\mathcal{X}_s,x'_s)
  \le
  \frac1n\sum_{s=1}^{n} w_Q(\mathcal{X}_s) Z_{f^{\star}}(\mathcal{X}_s,x'_s).
\tag{E.4}\label{eq:opt}
\]
Apply \eqref{eq:uniform} separately with $h=\widehat f$ and $h=f^{\star}$, then subtract the two inequalities:

\[
\begin{aligned}
  \mathbb E_{\nu_Q}[w_Q Z_{\widehat f}]
  -\mathbb E_{\nu_Q}[w_Q Z_{f^{\star}}]
  &= 
  \Bigl(
    \mathbb E_{\nu_Q}[w_Q Z_{\widehat f}]
    -\tfrac1n\ \sum_{s}w_Q(\mathcal{X}_s) Z_{\widehat f}
  \Bigr)
  -
  \Bigl(
    \mathbb E_{\nu_Q}[w_Q Z_{f^{\star}}]
    -\tfrac1n\ \sum_{s}w_Q(\mathcal{X}_s) Z_{f^{\star}}
  \Bigr) \\
  &\quad
  +\;
  \Bigl(
    \tfrac1n\ \sum_{s}w_Q(\mathcal{X}_s) Z_{\widehat f}
    -\tfrac1n\ \sum_{s}w_Q(\mathcal{X}_s) Z_{f^{\star}}
  \Bigr).
\end{aligned}
\]

The first two brackets are each bounded by
\(\pm C L\Delta_Q \sqrt{\kappa_Q d \log(1/\delta)/n}\)
from \eqref{eq:uniform};
the last bracket is non-positive by \eqref{eq:opt}.
Hence

\[
  \mathbb E_{\nu_Q}[w_Q Z_{\widehat f}]
  -\mathbb E_{\nu_Q}[w_Q Z_{f^{\star}}]
  \le
  2C L\Delta_Q
     \sqrt{\frac{\kappa_Q\,d\,\log(1/\delta)}{n}} .
\tag{E.5}\label{eq:cancellation}
\]

Absorb the factor \(2\) into \(C\).

\emph{Step 5: combine with the remaining terms.}
Replace the student–teacher gap in the risk decomposition by
\eqref{eq:cancellation};
retain the Lipschitz–geometry bound
\(L\Delta_Q\) and the teacher–Bayes term
\(\eta\bigl(H(g)\bigr)\) from the unbiased case.
This yields \eqref{eq:bias-bound}.
\end{proof}

\section{Additional Experimental Details}
\label{app:details}
\subsection{Dataset Descriptions}
\label{app:datasets}

Table \ref{tab:datasets} describes all test collections in terms of their domain, queries and corpus size. In all cases we rerank BM25 ($k_1=1.2$, $b=0.75$). We found in further experimentation that due to the point-wise nature of models trained in this investigation, biases remained consistent under different re-rankers thus for conciseness we solely present BM25. 

\begin{table}[htb]
    \centering
    \caption{Descriptive statistics for all test collections, $|\mathcal{Q}|$ indicates the number of test queries, $|\mathcal{D}|$ indicates the corpus size used in retrieval and ranking. In all cases, we re-rank BM25.}
    \begin{tabular}{llrr}
    \toprule
        Dataset & Domain & |$\mathcal{Q}$| & |$\mathcal{D}$|  \\
    \midrule
        TREC Deep-Learning 2019 & Ad-Hoc Web Search & 43 & \multirow{2}{*}{8E6} \\
        TREC Deep-Learning 2020 & Ad-Hoc Web Search & 53 & \\
    \midrule
        ArguAna & Argument Retrieval & 1406 & 8.67E3 \\
        Climate-Fever & Environmental & 1535 & 542E3 \\
        CQADupStack & OpenQA & 13145 & 457E3 \\
        DBPedia & OpenQA & 400 & 463E3 \\
        FiQA & OpenQA & 648 & 57E3 \\
        HotpotQA & OpenQA & 7405 & 523E3 \\
        NFCorpus & Medical & 323 & 36E2 \\
        NQ & OpenQA & 3452 & 268E3\\
        Quora & OpenQA & 10000 & 523E3 \\
        SCIDOCS & Academic & 1000 & 25E3 \\
        SciFact & Academic & 300 & 5E3 \\
        TREC Covid & Medical & 50 & 171E3 \\
        Touche 2020 & Argument Retrieval & 49 & 382E3 \\
    \bottomrule
    \end{tabular}
    
    \label{tab:datasets}
\end{table}

\subsection{Empirical Approximations}
We apply Monte-Carlo sampling to provide empirical estimates of theoretical values outlined in Section \ref{sec:theory}, our sample size is equal to our training corpus (maximising possible samples under our setting). We compute $H_\nu(g)$ over teacher scores of training data observed during the training of each model.

As our true measure $\mu_Q$ over $\mathcal{X}$ is latent (or infeasible to compute with standard benchmarks due to query mismatch), we provide an approximation of $\kappa_Q$, the density ratio bridging our biased measure $\nu_Q$ to our unbiased risk. We do so by assuming a uniform chance of sampling all documents as opposed to a rank-biased sample taking $1/g(Q, D), \forall D \in \mathcal{X}_Q$, $\kappa_Q$ is then taken as the supremum of these values. 

To compute an empirical diameter $\widehat \Delta_Q$, we apply cosine distance over representation from an existing embedding model (we use RetroMAE~\citep{xiao:2022}, a strong embedding model based on MAE pre-training) as our measure $d$ and compute the supremum over Monte-Carlo samples from our training data over each query.

\subsection{TOST test}

A two one-sided t-test (TOST) determines if the means of two populations are equivalent based on independent samples from each population, in our case, the query-level effectiveness of two models. For means $\mu_1, \mu_2$ and confidence bound $\theta=|\mu_2-\mu_1|\cdot \epsilon$ ($\epsilon$ is a percentage bound parameter), we assess two hypotheses, that  $\mu_2-\mu_1$ lies above $\theta$ and below $\theta$ using one-sided t-tests with confidence $1-2\alpha$ compensating for multiple hypothesis testing. Thus, within an $\epsilon$ bound with confidence $1-\alpha$, we can say that $\mu_1, \mu_2$ are statistically equivalent.

\section{Out-of-domain Effectiveness}
\label{app:ood}

Table \ref{tab:ood-be} shows all BEIR splits for bi-encoder models. Table \ref{tab:ood-ce} shows all splits for cross-encoder models.

\begin{table}[htb]
  \caption{Mean nDCG@10 for architecture BE across BEIR datasets}
  \centering
  \footnotesize
  \begin{adjustbox}{width=\textwidth}
  \renewcommand{\arraystretch}{0.8}
  \setlength{\tabcolsep}{1pt}
  \begin{tabular}{llcccccccccccccccccccccccc}
  \toprule
    Loss & Domain & arguana & climate-fever & cqa-android & cqa-english & cqa-gaming & cqa-gis & cqa-mathematica & cqa-physics & cqa-programmers & cqa-stats & cqa-tex & cqa-unix & cqa-webmasters & cqa-wordpress & dbpedia-entity & fiqa & hotpotqa & nfcorpus & nq & quora & scidocs & scifact & trec-covid & webis-touche2020 \\
  \midrule
LCE & Random & 0.414\textsuperscript{} & 0.234\textsuperscript{} & 0.319\textsuperscript{BCD} & 0.295\textsuperscript{C} & 0.393\textsuperscript{BC} & 0.220\textsuperscript{BC} & 0.169\textsuperscript{BC} & 0.308\textsuperscript{} & 0.254\textsuperscript{BC} & 0.194\textsuperscript{BCD} & 0.182\textsuperscript{} & 0.230\textsuperscript{C} & 0.244\textsuperscript{BC} & 0.195\textsuperscript{BCD} & 0.318\textsuperscript{BC} & 0.241\textsuperscript{} & 0.547\textsuperscript{} & 0.296\textsuperscript{BCD} & 0.322\textsuperscript{} & 0.810\textsuperscript{} & 0.130\textsuperscript{} & 0.534\textsuperscript{BCD} & 0.628\textsuperscript{BCD} & 0.291\textsuperscript{BCD} \\
LCE & BM25 & 0.303\textsuperscript{} & 0.193\textsuperscript{} & 0.333\textsuperscript{ACD} & 0.307\textsuperscript{} & 0.398\textsuperscript{AC} & 0.215\textsuperscript{AC} & 0.163\textsuperscript{AC} & 0.284\textsuperscript{C} & 0.267\textsuperscript{AC} & 0.201\textsuperscript{ACD} & 0.201\textsuperscript{} & 0.241\textsuperscript{} & 0.254\textsuperscript{AC} & 0.217\textsuperscript{ACD} & 0.321\textsuperscript{AC} & 0.259\textsuperscript{CD} & 0.559\textsuperscript{} & 0.291\textsuperscript{ACD} & 0.407\textsuperscript{D} & 0.798\textsuperscript{} & 0.110\textsuperscript{} & 0.468\textsuperscript{ACD} & 0.671\textsuperscript{ACD} & 0.296\textsuperscript{ACD} \\
LCE & CE & 0.281\textsuperscript{} & 0.191\textsuperscript{} & 0.317\textsuperscript{ABD} & 0.292\textsuperscript{A} & 0.392\textsuperscript{AB} & 0.204\textsuperscript{AB} & 0.154\textsuperscript{AB} & 0.282\textsuperscript{B} & 0.255\textsuperscript{AB} & 0.191\textsuperscript{ABD} & 0.193\textsuperscript{} & 0.231\textsuperscript{A} & 0.260\textsuperscript{AB} & 0.205\textsuperscript{ABD} & 0.324\textsuperscript{AB} & 0.261\textsuperscript{BD} & 0.555\textsuperscript{} & 0.283\textsuperscript{ABD} & 0.405\textsuperscript{} & 0.791\textsuperscript{} & 0.104\textsuperscript{} & 0.491\textsuperscript{ABD} & 0.673\textsuperscript{ABD} & 0.284\textsuperscript{ABD} \\
LCE & Ensemble & 0.352\textsuperscript{} & 0.199\textsuperscript{} & 0.340\textsuperscript{ABC} & 0.327\textsuperscript{} & 0.427\textsuperscript{} & 0.237\textsuperscript{} & 0.183\textsuperscript{} & 0.321\textsuperscript{} & 0.282\textsuperscript{} & 0.208\textsuperscript{ABC} & 0.214\textsuperscript{} & 0.255\textsuperscript{} & 0.284\textsuperscript{} & 0.227\textsuperscript{ABC} & 0.346\textsuperscript{} & 0.272\textsuperscript{BC} & – & 0.303\textsuperscript{ABC} & 0.407\textsuperscript{B} & 0.823\textsuperscript{} & 0.120\textsuperscript{} & 0.507\textsuperscript{ABC} & 0.685\textsuperscript{ABC} & 0.336\textsuperscript{ABC} \\
RankNet & Random & 0.181\textsuperscript{} & 0.086\textsuperscript{} & 0.144\textsuperscript{} & 0.145\textsuperscript{} & 0.125\textsuperscript{} & 0.104\textsuperscript{} & 0.100\textsuperscript{} & 0.151\textsuperscript{} & 0.145\textsuperscript{} & 0.105\textsuperscript{} & 0.099\textsuperscript{} & 0.119\textsuperscript{} & 0.120\textsuperscript{} & 0.100\textsuperscript{} & 0.247\textsuperscript{} & 0.142\textsuperscript{} & 0.335\textsuperscript{} & 0.287\textsuperscript{BCD} & 0.238\textsuperscript{} & 0.707\textsuperscript{} & 0.097\textsuperscript{} & 0.259\textsuperscript{} & 0.593\textsuperscript{BCD} & 0.237\textsuperscript{} \\
RankNet & BM25 & 0.308\textsuperscript{} & 0.229\textsuperscript{} & 0.313\textsuperscript{} & 0.299\textsuperscript{} & 0.384\textsuperscript{} & 0.205\textsuperscript{} & 0.156\textsuperscript{} & 0.278\textsuperscript{} & 0.254\textsuperscript{C} & 0.195\textsuperscript{C} & 0.179\textsuperscript{} & 0.232\textsuperscript{C} & 0.243\textsuperscript{C} & 0.194\textsuperscript{C} & 0.354\textsuperscript{CD} & 0.274\textsuperscript{CD} & 0.592\textsuperscript{} & 0.293\textsuperscript{ACD} & 0.415\textsuperscript{D} & 0.802\textsuperscript{} & 0.115\textsuperscript{} & 0.509\textsuperscript{C} & 0.657\textsuperscript{ACD} & 0.306\textsuperscript{CD} \\
RankNet & CE & 0.294\textsuperscript{} & 0.215\textsuperscript{} & 0.333\textsuperscript{D} & 0.310\textsuperscript{} & 0.402\textsuperscript{} & 0.219\textsuperscript{} & 0.168\textsuperscript{} & 0.291\textsuperscript{} & 0.256\textsuperscript{B} & 0.201\textsuperscript{B} & 0.188\textsuperscript{} & 0.238\textsuperscript{B} & 0.249\textsuperscript{B} & 0.207\textsuperscript{B} & 0.351\textsuperscript{BD} & 0.272\textsuperscript{BD} & 0.604\textsuperscript{} & 0.296\textsuperscript{ABD} & 0.420\textsuperscript{} & 0.807\textsuperscript{} & 0.113\textsuperscript{} & 0.502\textsuperscript{B} & 0.646\textsuperscript{ABD} & 0.323\textsuperscript{BD} \\
RankNet & Ensemble & 0.344\textsuperscript{} & 0.256\textsuperscript{} & 0.348\textsuperscript{C} & 0.338\textsuperscript{} & 0.444\textsuperscript{} & 0.251\textsuperscript{} & 0.185\textsuperscript{} & 0.324\textsuperscript{} & 0.281\textsuperscript{} & 0.224\textsuperscript{} & 0.212\textsuperscript{} & 0.265\textsuperscript{} & 0.274\textsuperscript{} & 0.226\textsuperscript{} & 0.363\textsuperscript{BC} & 0.282\textsuperscript{BC} & 0.599\textsuperscript{} & 0.310\textsuperscript{ABC} & 0.416\textsuperscript{B} & 0.835\textsuperscript{} & 0.131\textsuperscript{} & 0.554\textsuperscript{} & 0.681\textsuperscript{ABC} & 0.357\textsuperscript{BC} \\
mMSE & Random & 0.382\textsuperscript{} & 0.218\textsuperscript{D} & 0.324\textsuperscript{C} & 0.309\textsuperscript{C} & 0.393\textsuperscript{C} & 0.199\textsuperscript{BC} & 0.154\textsuperscript{BC} & 0.304\textsuperscript{D} & 0.257\textsuperscript{BC} & 0.182\textsuperscript{} & 0.172\textsuperscript{} & 0.231\textsuperscript{BC} & 0.251\textsuperscript{BC} & 0.184\textsuperscript{} & 0.351\textsuperscript{D} & 0.241\textsuperscript{} & 0.525\textsuperscript{} & 0.299\textsuperscript{BCD} & 0.384\textsuperscript{} & 0.818\textsuperscript{} & 0.117\textsuperscript{} & 0.498\textsuperscript{BCD} & 0.679\textsuperscript{BCD} & 0.314\textsuperscript{BCD} \\
mMSE & BM25 & 0.299\textsuperscript{} & 0.193\textsuperscript{} & 0.293\textsuperscript{} & 0.295\textsuperscript{} & 0.363\textsuperscript{} & 0.201\textsuperscript{AC} & 0.156\textsuperscript{AC} & 0.270\textsuperscript{} & 0.254\textsuperscript{AC} & 0.199\textsuperscript{CD} & 0.184\textsuperscript{} & 0.229\textsuperscript{AC} & 0.239\textsuperscript{AC} & 0.202\textsuperscript{C} & 0.335\textsuperscript{C} & 0.262\textsuperscript{C} & 0.587\textsuperscript{} & 0.288\textsuperscript{ACD} & 0.414\textsuperscript{C} & 0.816\textsuperscript{} & 0.111\textsuperscript{} & 0.496\textsuperscript{ACD} & 0.662\textsuperscript{ACD} & 0.309\textsuperscript{ACD} \\
mMSE & CE & 0.312\textsuperscript{} & 0.190\textsuperscript{} & 0.323\textsuperscript{A} & 0.306\textsuperscript{A} & 0.402\textsuperscript{A} & 0.210\textsuperscript{AB} & 0.161\textsuperscript{AB} & 0.286\textsuperscript{} & 0.261\textsuperscript{AB} & 0.207\textsuperscript{BD} & 0.194\textsuperscript{} & 0.236\textsuperscript{AB} & 0.252\textsuperscript{AB} & 0.209\textsuperscript{B} & 0.332\textsuperscript{B} & 0.262\textsuperscript{B} & 0.591\textsuperscript{} & 0.293\textsuperscript{ABD} & 0.416\textsuperscript{B} & 0.816\textsuperscript{} & 0.108\textsuperscript{} & 0.504\textsuperscript{ABD} & 0.655\textsuperscript{ABD} & 0.313\textsuperscript{ABD} \\
mMSE & Ensemble & 0.366\textsuperscript{} & 0.219\textsuperscript{A} & 0.343\textsuperscript{} & 0.328\textsuperscript{} & 0.429\textsuperscript{} & 0.245\textsuperscript{} & 0.188\textsuperscript{} & 0.314\textsuperscript{A} & 0.282\textsuperscript{} & 0.216\textsuperscript{BC} & 0.217\textsuperscript{} & 0.262\textsuperscript{} & 0.284\textsuperscript{} & 0.231\textsuperscript{} & 0.351\textsuperscript{A} & 0.276\textsuperscript{} & 0.599\textsuperscript{} & 0.304\textsuperscript{ABC} & 0.407\textsuperscript{} & 0.840\textsuperscript{} & 0.121\textsuperscript{} & 0.537\textsuperscript{ABC} & 0.679\textsuperscript{ABC} & 0.349\textsuperscript{ABC} \\
KL & Random & 0.346\textsuperscript{} & 0.232\textsuperscript{} & 0.312\textsuperscript{BCD} & 0.308\textsuperscript{BC} & 0.405\textsuperscript{BC} & 0.218\textsuperscript{BC} & 0.164\textsuperscript{BC} & 0.302\textsuperscript{BC} & 0.259\textsuperscript{BC} & 0.197\textsuperscript{BCD} & 0.179\textsuperscript{} & 0.234\textsuperscript{BC} & 0.246\textsuperscript{BCD} & 0.201\textsuperscript{BCD} & 0.334\textsuperscript{BCD} & 0.262\textsuperscript{BCD} & 0.559\textsuperscript{} & 0.303\textsuperscript{BCD} & 0.359\textsuperscript{} & 0.745\textsuperscript{} & 0.132\textsuperscript{} & 0.539\textsuperscript{CD} & 0.665\textsuperscript{BCD} & 0.298\textsuperscript{BCD} \\
KL & BM25 & 0.324\textsuperscript{C} & 0.196\textsuperscript{} & 0.334\textsuperscript{ACD} & 0.315\textsuperscript{AC} & 0.403\textsuperscript{AC} & 0.214\textsuperscript{AC} & 0.163\textsuperscript{AC} & 0.291\textsuperscript{AC} & 0.267\textsuperscript{AC} & 0.202\textsuperscript{ACD} & 0.195\textsuperscript{} & 0.241\textsuperscript{AC} & 0.253\textsuperscript{ACD} & 0.215\textsuperscript{ACD} & 0.339\textsuperscript{ACD} & 0.275\textsuperscript{ACD} & 0.559\textsuperscript{} & 0.286\textsuperscript{ACD} & 0.414\textsuperscript{CD} & 0.806\textsuperscript{} & 0.109\textsuperscript{} & 0.466\textsuperscript{} & 0.687\textsuperscript{ACD} & 0.322\textsuperscript{ACD} \\
KL & CE & 0.321\textsuperscript{B} & 0.207\textsuperscript{D} & 0.324\textsuperscript{ABD} & 0.303\textsuperscript{AB} & 0.407\textsuperscript{AB} & 0.209\textsuperscript{AB} & 0.160\textsuperscript{AB} & 0.301\textsuperscript{AB} & 0.258\textsuperscript{AB} & 0.204\textsuperscript{ABD} & 0.193\textsuperscript{} & 0.237\textsuperscript{AB} & 0.265\textsuperscript{ABD} & 0.210\textsuperscript{ABD} & 0.340\textsuperscript{ABD} & 0.268\textsuperscript{ABD} & 0.565\textsuperscript{} & 0.289\textsuperscript{ABD} & 0.413\textsuperscript{BD} & 0.805\textsuperscript{} & 0.111\textsuperscript{} & 0.501\textsuperscript{AD} & 0.670\textsuperscript{ABD} & 0.323\textsuperscript{ABD} \\
KL & Ensemble & 0.354\textsuperscript{} & 0.206\textsuperscript{C} & 0.339\textsuperscript{ABC} & 0.333\textsuperscript{} & 0.431\textsuperscript{} & 0.238\textsuperscript{} & 0.184\textsuperscript{} & 0.321\textsuperscript{} & 0.281\textsuperscript{} & 0.216\textsuperscript{ABC} & 0.213\textsuperscript{} & 0.258\textsuperscript{} & 0.283\textsuperscript{ABC} & 0.228\textsuperscript{ABC} & 0.350\textsuperscript{ABC} & 0.277\textsuperscript{ABC} & 0.585\textsuperscript{} & 0.306\textsuperscript{ABC} & 0.413\textsuperscript{BC} & 0.827\textsuperscript{} & 0.120\textsuperscript{} & 0.519\textsuperscript{AC} & 0.669\textsuperscript{ABC} & 0.344\textsuperscript{ABC} \\
  \bottomrule
  \end{tabular}
  \end{adjustbox}
  \label{tab:ood-be}
\end{table}
\begin{table}[htb]
  \caption{Mean nDCG@10 for architecture CE across BEIR datasets}
  \centering
  \footnotesize
  \begin{adjustbox}{width=\textwidth}
  \renewcommand{\arraystretch}{0.8}
  \begin{tabular}{llcccccccccccccccccccccccc}
  \toprule
    Loss & Domain & arguana & climate-fever & cqa-android & cqa-english & cqa-gaming & cqa-gis & cqa-mathematica & cqa-physics & cqa-programmers & cqa-stats & cqa-tex & cqa-unix & cqa-webmasters & cqa-wordpress & dbpedia-entity & fiqa & hotpotqa & nfcorpus & nq & quora & scidocs & scifact & trec-covid & webis-touche2020 \\
  \midrule
LCE & Random & 0.311\textsuperscript{} & 0.217\textsuperscript{} & 0.386\textsuperscript{D} & 0.373\textsuperscript{} & 0.475\textsuperscript{} & 0.292\textsuperscript{} & 0.216\textsuperscript{BCD} & 0.335\textsuperscript{D} & 0.306\textsuperscript{D} & 0.265\textsuperscript{D} & 0.245\textsuperscript{D} & 0.288\textsuperscript{BCD} & 0.329\textsuperscript{} & 0.276\textsuperscript{BCD} & 0.369\textsuperscript{} & 0.317\textsuperscript{C} & 0.680\textsuperscript{} & 0.340\textsuperscript{BCD} & 0.384\textsuperscript{} & 0.795\textsuperscript{} & 0.156\textsuperscript{} & 0.665\textsuperscript{BD} & 0.658\textsuperscript{BCD} & 0.340\textsuperscript{BCD} \\
LCE & BM25 & 0.238\textsuperscript{} & 0.205\textsuperscript{} & 0.356\textsuperscript{C} & 0.318\textsuperscript{} & 0.436\textsuperscript{} & 0.253\textsuperscript{C} & 0.208\textsuperscript{ACD} & 0.305\textsuperscript{} & 0.266\textsuperscript{C} & 0.232\textsuperscript{C} & 0.226\textsuperscript{C} & 0.288\textsuperscript{ACD} & 0.284\textsuperscript{CD} & 0.249\textsuperscript{ACD} & 0.404\textsuperscript{CD} & 0.361\textsuperscript{D} & 0.686\textsuperscript{} & 0.328\textsuperscript{ACD} & 0.474\textsuperscript{CD} & 0.670\textsuperscript{} & 0.147\textsuperscript{} & 0.643\textsuperscript{AD} & 0.728\textsuperscript{ACD} & 0.339\textsuperscript{ACD} \\
LCE & CE & 0.250\textsuperscript{} & 0.190\textsuperscript{} & 0.363\textsuperscript{B} & 0.298\textsuperscript{} & 0.447\textsuperscript{D} & 0.261\textsuperscript{B} & 0.198\textsuperscript{ABD} & 0.287\textsuperscript{} & 0.267\textsuperscript{B} & 0.230\textsuperscript{B} & 0.226\textsuperscript{B} & 0.289\textsuperscript{ABD} & 0.288\textsuperscript{BD} & 0.263\textsuperscript{ABD} & 0.406\textsuperscript{BD} & 0.320\textsuperscript{A} & 0.684\textsuperscript{} & 0.311\textsuperscript{ABD} & 0.477\textsuperscript{BD} & 0.703\textsuperscript{} & 0.135\textsuperscript{} & 0.582\textsuperscript{} & 0.716\textsuperscript{ABD} & 0.327\textsuperscript{ABD} \\
LCE & Ensemble & 0.352\textsuperscript{} & 0.164\textsuperscript{} & 0.381\textsuperscript{A} & 0.349\textsuperscript{} & 0.453\textsuperscript{C} & 0.276\textsuperscript{} & 0.217\textsuperscript{ABC} & 0.326\textsuperscript{A} & 0.297\textsuperscript{A} & 0.260\textsuperscript{A} & 0.247\textsuperscript{A} & 0.288\textsuperscript{ABC} & 0.303\textsuperscript{BC} & 0.272\textsuperscript{ABC} & 0.414\textsuperscript{BC} & 0.360\textsuperscript{B} & 0.694\textsuperscript{} & 0.327\textsuperscript{ABC} & 0.474\textsuperscript{BC} & 0.746\textsuperscript{} & 0.151\textsuperscript{} & 0.661\textsuperscript{AB} & 0.700\textsuperscript{ABC} & 0.361\textsuperscript{ABC} \\
RankNet & Random & 0.363\textsuperscript{BCD} & 0.177\textsuperscript{} & 0.379\textsuperscript{BCD} & 0.376\textsuperscript{BCD} & 0.461\textsuperscript{BCD} & 0.280\textsuperscript{BCD} & 0.216\textsuperscript{BCD} & 0.334\textsuperscript{} & 0.309\textsuperscript{BCD} & 0.254\textsuperscript{BCD} & 0.230\textsuperscript{} & 0.291\textsuperscript{} & 0.324\textsuperscript{BCD} & 0.240\textsuperscript{} & 0.309\textsuperscript{} & 0.307\textsuperscript{} & 0.598\textsuperscript{} & 0.338\textsuperscript{BCD} & 0.353\textsuperscript{} & 0.633\textsuperscript{} & 0.149\textsuperscript{} & 0.672\textsuperscript{BCD} & 0.628\textsuperscript{} & 0.313\textsuperscript{} \\
RankNet & BM25 & 0.362\textsuperscript{ACD} & 0.260\textsuperscript{} & 0.379\textsuperscript{ACD} & 0.381\textsuperscript{ACD} & 0.474\textsuperscript{ACD} & 0.291\textsuperscript{ACD} & 0.219\textsuperscript{ACD} & 0.348\textsuperscript{CD} & 0.314\textsuperscript{ACD} & 0.265\textsuperscript{ACD} & 0.249\textsuperscript{CD} & 0.304\textsuperscript{CD} & 0.314\textsuperscript{ACD} & 0.277\textsuperscript{CD} & 0.422\textsuperscript{CD} & 0.368\textsuperscript{CD} & 0.720\textsuperscript{} & 0.342\textsuperscript{ACD} & 0.482\textsuperscript{C} & 0.785\textsuperscript{} & 0.154\textsuperscript{} & 0.701\textsuperscript{ACD} & 0.707\textsuperscript{CD} & 0.375\textsuperscript{CD} \\
RankNet & CE & 0.369\textsuperscript{ABD} & 0.253\textsuperscript{D} & 0.381\textsuperscript{ABD} & 0.376\textsuperscript{ABD} & 0.473\textsuperscript{ABD} & 0.288\textsuperscript{ABD} & 0.219\textsuperscript{ABD} & 0.351\textsuperscript{BD} & 0.313\textsuperscript{ABD} & 0.272\textsuperscript{ABD} & 0.247\textsuperscript{BD} & 0.304\textsuperscript{BD} & 0.313\textsuperscript{ABD} & 0.280\textsuperscript{BD} & 0.419\textsuperscript{BD} & 0.372\textsuperscript{BD} & 0.717\textsuperscript{} & 0.339\textsuperscript{ABD} & 0.482\textsuperscript{B} & 0.787\textsuperscript{} & 0.152\textsuperscript{D} & 0.701\textsuperscript{ABD} & 0.695\textsuperscript{BD} & 0.367\textsuperscript{BD} \\
RankNet & Ensemble & 0.364\textsuperscript{ABC} & 0.252\textsuperscript{C} & 0.383\textsuperscript{ABC} & 0.377\textsuperscript{ABC} & 0.469\textsuperscript{ABC} & 0.290\textsuperscript{ABC} & 0.220\textsuperscript{ABC} & 0.347\textsuperscript{BC} & 0.311\textsuperscript{ABC} & 0.265\textsuperscript{ABC} & 0.250\textsuperscript{BC} & 0.304\textsuperscript{BC} & 0.308\textsuperscript{ABC} & 0.279\textsuperscript{BC} & 0.423\textsuperscript{BC} & 0.369\textsuperscript{BC} & – & 0.342\textsuperscript{ABC} & 0.476\textsuperscript{} & 0.810\textsuperscript{} & 0.151\textsuperscript{C} & 0.719\textsuperscript{ABC} & 0.688\textsuperscript{BC} & 0.367\textsuperscript{BC} \\
mMSE & Random & 0.305\textsuperscript{} & 0.182\textsuperscript{} & 0.368\textsuperscript{BCD} & 0.365\textsuperscript{BCD} & 0.460\textsuperscript{BCD} & 0.276\textsuperscript{BCD} & 0.201\textsuperscript{} & 0.333\textsuperscript{BCD} & 0.303\textsuperscript{BCD} & 0.268\textsuperscript{BCD} & 0.237\textsuperscript{} & 0.284\textsuperscript{} & 0.305\textsuperscript{BCD} & 0.259\textsuperscript{BCD} & 0.403\textsuperscript{BCD} & 0.338\textsuperscript{} & 0.655\textsuperscript{} & 0.343\textsuperscript{BCD} & 0.451\textsuperscript{} & 0.691\textsuperscript{} & 0.148\textsuperscript{} & 0.680\textsuperscript{BCD} & 0.713\textsuperscript{BCD} & 0.340\textsuperscript{BCD} \\
mMSE & BM25 & 0.348\textsuperscript{} & 0.249\textsuperscript{} & 0.387\textsuperscript{ACD} & 0.363\textsuperscript{ACD} & 0.465\textsuperscript{ACD} & 0.282\textsuperscript{ACD} & 0.215\textsuperscript{CD} & 0.334\textsuperscript{ACD} & 0.299\textsuperscript{ACD} & 0.263\textsuperscript{ACD} & 0.243\textsuperscript{} & 0.303\textsuperscript{CD} & 0.301\textsuperscript{ACD} & 0.269\textsuperscript{ACD} & 0.421\textsuperscript{ACD} & 0.374\textsuperscript{CD} & 0.707\textsuperscript{} & 0.342\textsuperscript{ACD} & 0.485\textsuperscript{C} & 0.782\textsuperscript{} & 0.150\textsuperscript{} & 0.686\textsuperscript{ACD} & 0.726\textsuperscript{ACD} & 0.366\textsuperscript{ACD} \\
mMSE & CE & 0.357\textsuperscript{} & 0.240\textsuperscript{} & 0.390\textsuperscript{ABD} & 0.371\textsuperscript{ABD} & 0.470\textsuperscript{ABD} & 0.289\textsuperscript{ABD} & 0.222\textsuperscript{BD} & 0.340\textsuperscript{ABD} & 0.309\textsuperscript{ABD} & 0.261\textsuperscript{ABD} & 0.247\textsuperscript{D} & 0.306\textsuperscript{BD} & 0.311\textsuperscript{ABD} & 0.275\textsuperscript{ABD} & 0.416\textsuperscript{ABD} & 0.376\textsuperscript{BD} & 0.706\textsuperscript{} & 0.346\textsuperscript{ABD} & 0.485\textsuperscript{B} & 0.788\textsuperscript{} & 0.154\textsuperscript{D} & 0.703\textsuperscript{ABD} & 0.714\textsuperscript{ABD} & 0.356\textsuperscript{ABD} \\
mMSE & Ensemble & 0.367\textsuperscript{} & 0.217\textsuperscript{} & 0.384\textsuperscript{ABC} & 0.374\textsuperscript{ABC} & 0.457\textsuperscript{ABC} & 0.285\textsuperscript{ABC} & 0.218\textsuperscript{BC} & 0.334\textsuperscript{ABC} & 0.309\textsuperscript{ABC} & 0.268\textsuperscript{ABC} & 0.246\textsuperscript{C} & 0.297\textsuperscript{BC} & 0.305\textsuperscript{ABC} & 0.277\textsuperscript{ABC} & 0.423\textsuperscript{ABC} & 0.376\textsuperscript{BC} & 0.708\textsuperscript{} & 0.342\textsuperscript{ABC} & 0.480\textsuperscript{} & 0.793\textsuperscript{} & 0.155\textsuperscript{C} & 0.699\textsuperscript{ABC} & 0.693\textsuperscript{ABC} & 0.379\textsuperscript{ABC} \\
KL & Random & 0.342\textsuperscript{} & 0.180\textsuperscript{} & 0.403\textsuperscript{BCD} & 0.375\textsuperscript{} & 0.491\textsuperscript{} & 0.310\textsuperscript{} & 0.231\textsuperscript{BCD} & 0.351\textsuperscript{} & 0.318\textsuperscript{} & 0.275\textsuperscript{BCD} & 0.255\textsuperscript{} & 0.310\textsuperscript{BC} & 0.332\textsuperscript{} & 0.283\textsuperscript{BCD} & 0.383\textsuperscript{} & 0.324\textsuperscript{} & 0.630\textsuperscript{} & 0.340\textsuperscript{BCD} & 0.422\textsuperscript{} & 0.822\textsuperscript{} & 0.146\textsuperscript{} & 0.681\textsuperscript{BCD} & 0.678\textsuperscript{BCD} & 0.330\textsuperscript{BCD} \\
KL & BM25 & 0.287\textsuperscript{} & 0.204\textsuperscript{C} & 0.381\textsuperscript{ACD} & 0.331\textsuperscript{} & 0.463\textsuperscript{CD} & 0.276\textsuperscript{CD} & 0.214\textsuperscript{ACD} & 0.327\textsuperscript{CD} & 0.297\textsuperscript{CD} & 0.251\textsuperscript{ACD} & 0.240\textsuperscript{CD} & 0.299\textsuperscript{AC} & 0.303\textsuperscript{CD} & 0.271\textsuperscript{ACD} & 0.411\textsuperscript{CD} & 0.376\textsuperscript{CD} & 0.690\textsuperscript{} & 0.339\textsuperscript{ACD} & 0.482\textsuperscript{C} & 0.756\textsuperscript{} & 0.149\textsuperscript{D} & 0.682\textsuperscript{ACD} & 0.734\textsuperscript{ACD} & 0.341\textsuperscript{ACD} \\
KL & CE & 0.320\textsuperscript{} & 0.204\textsuperscript{B} & 0.387\textsuperscript{ABD} & 0.345\textsuperscript{D} & 0.469\textsuperscript{BD} & 0.278\textsuperscript{BD} & 0.222\textsuperscript{ABD} & 0.337\textsuperscript{BD} & 0.298\textsuperscript{BD} & 0.262\textsuperscript{ABD} & 0.243\textsuperscript{BD} & 0.307\textsuperscript{AB} & 0.304\textsuperscript{BD} & 0.278\textsuperscript{ABD} & 0.415\textsuperscript{BD} & 0.364\textsuperscript{BD} & 0.698\textsuperscript{} & 0.338\textsuperscript{ABD} & 0.481\textsuperscript{B} & 0.773\textsuperscript{} & 0.151\textsuperscript{} & 0.690\textsuperscript{ABD} & 0.704\textsuperscript{ABD} & 0.351\textsuperscript{ABD} \\
KL & Ensemble & 0.333\textsuperscript{} & 0.187\textsuperscript{} & 0.382\textsuperscript{ABC} & 0.347\textsuperscript{C} & 0.457\textsuperscript{BC} & 0.282\textsuperscript{BC} & 0.215\textsuperscript{ABC} & 0.323\textsuperscript{BC} & 0.292\textsuperscript{BC} & 0.257\textsuperscript{ABC} & 0.241\textsuperscript{BC} & 0.286\textsuperscript{} & 0.300\textsuperscript{BC} & 0.270\textsuperscript{ABC} & 0.416\textsuperscript{BC} & 0.367\textsuperscript{BC} & – & 0.341\textsuperscript{ABC} & 0.476\textsuperscript{} & 0.761\textsuperscript{} & 0.149\textsuperscript{B} & 0.671\textsuperscript{ABC} & 0.701\textsuperscript{ABC} & 0.395\textsuperscript{ABC} \\
  \bottomrule
  \end{tabular}
  \end{adjustbox}
  \label{tab:ood-ce}
\end{table}

\section{Additional Figures}
\label{app:figs}
\begin{figure}[tb]
  \centering
  \begin{subfigure}[b]{0.48\textwidth}
    \centering
    \includegraphics[width=\linewidth]{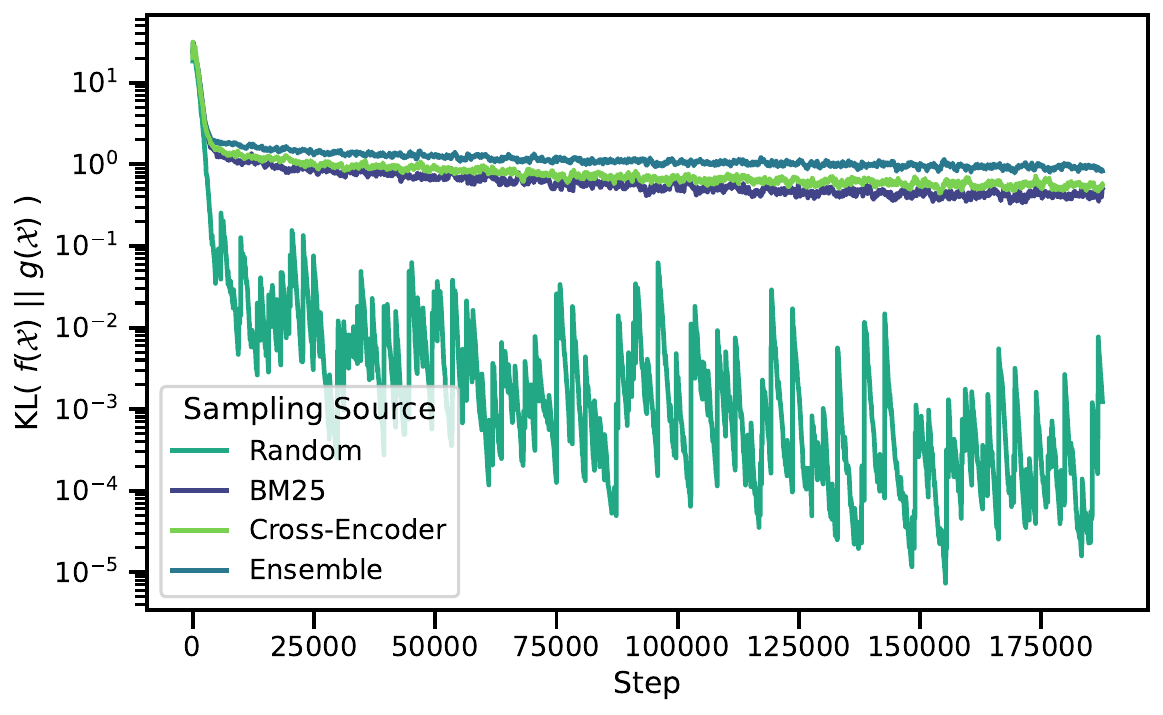}
    \caption{EMA of Loss}
    \label{fig:ema-loss}
  \end{subfigure}
  \hfill
  \begin{subfigure}[b]{0.48\textwidth}
    \centering
    \includegraphics[width=\linewidth]{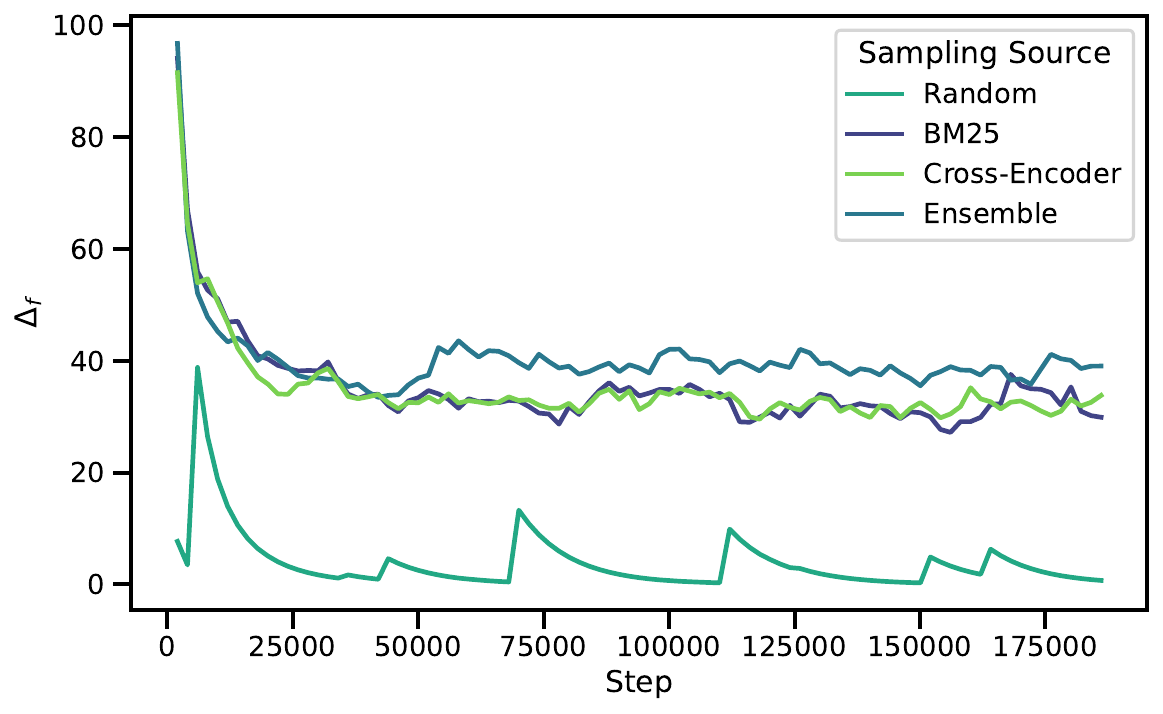}
    \caption{EMA of Grad Norm}
    \label{fig:ema-norm}
  \end{subfigure}
  \caption{Training loss in the form of the KL divergence (left) and gradient norm of the student $f$(right). Note the log scale of loss values. Observe the reduced variance in gradient under locality but otherwise minimal difference between sampling domains, we observe that loss converges marginally higher inversely with density ratio $\kappa_Q$.}
  \label{fig:ema-side-by-side}
\end{figure}

\begin{figure}[htp]
  \centering
  \begin{subfigure}[t]{0.24\textwidth}
    \centering
    \includegraphics[width=\linewidth]{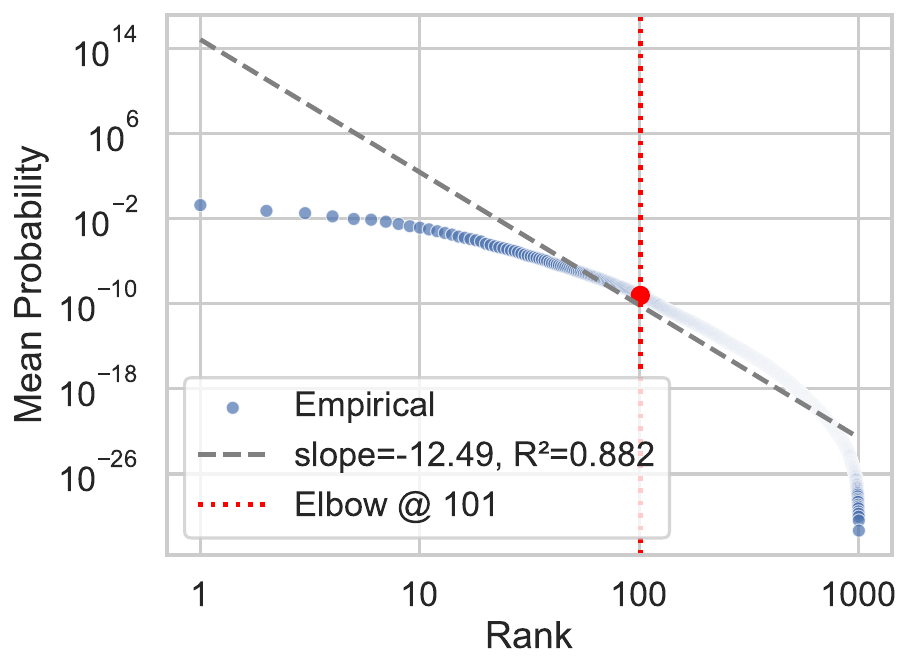}
    \caption*{\small LCE \\ Random}
  \end{subfigure}%
  \begin{subfigure}[t]{0.24\textwidth}
    \centering
    \includegraphics[width=\linewidth]{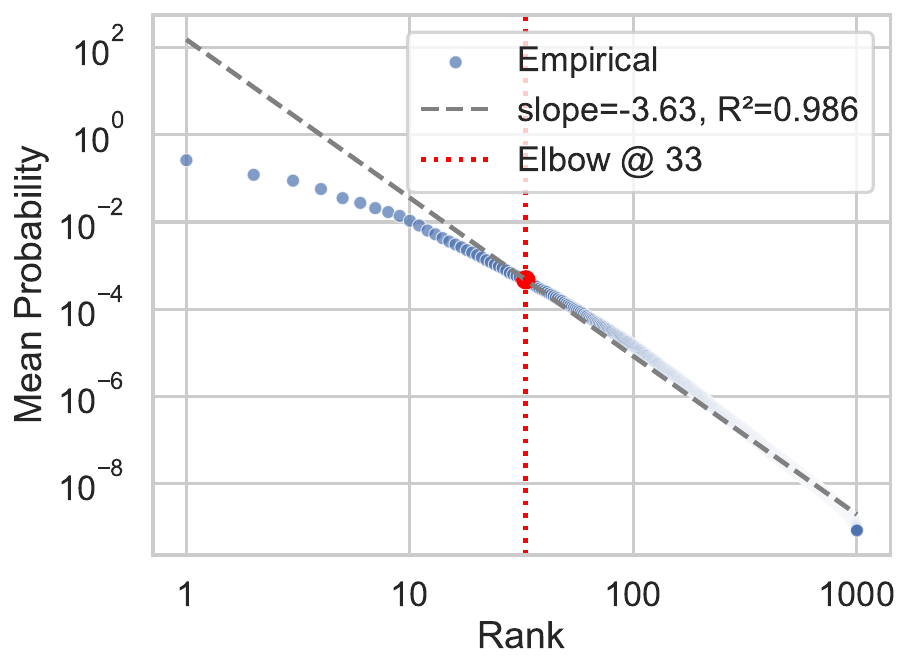}
    \caption*{\small LCE \\ BM25}
  \end{subfigure}%
  \begin{subfigure}[t]{0.24\textwidth}
    \centering
    \includegraphics[width=\linewidth]{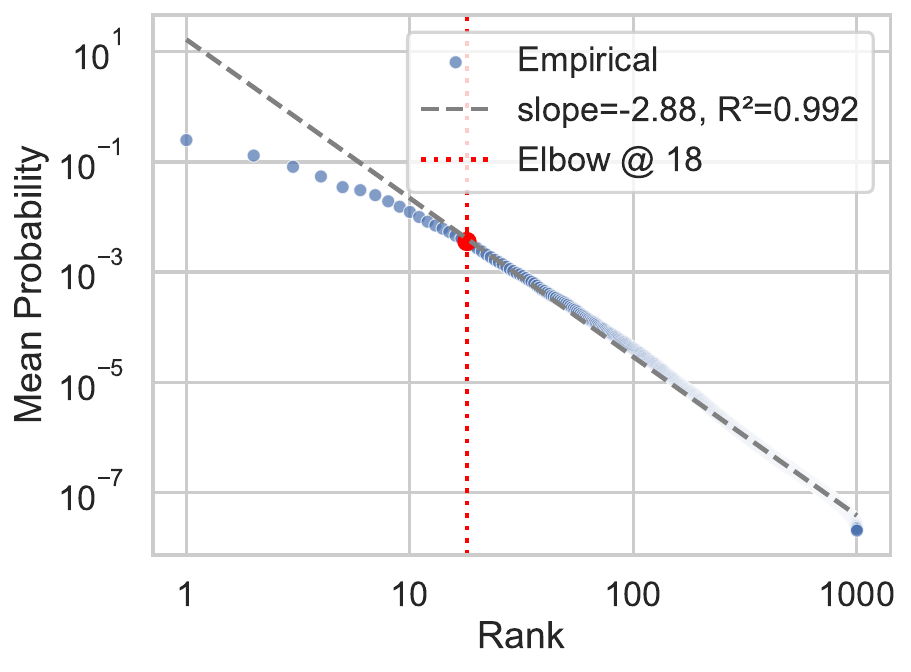}
    \caption*{\small LCE \\ Cross-Encoder}
  \end{subfigure}%
  \begin{subfigure}[t]{0.24\textwidth}
    \centering
    \includegraphics[width=\linewidth]{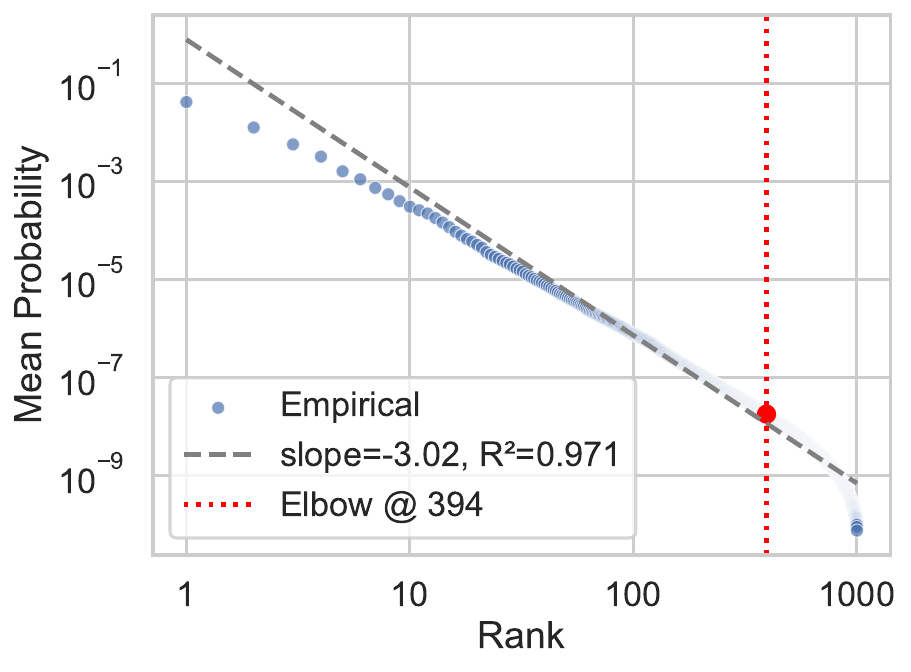}
    \caption*{\small LCE \\ Ensemble}
  \end{subfigure}

  \medskip
  \begin{subfigure}[t]{0.24\textwidth}
    \centering
    \includegraphics[width=\linewidth]{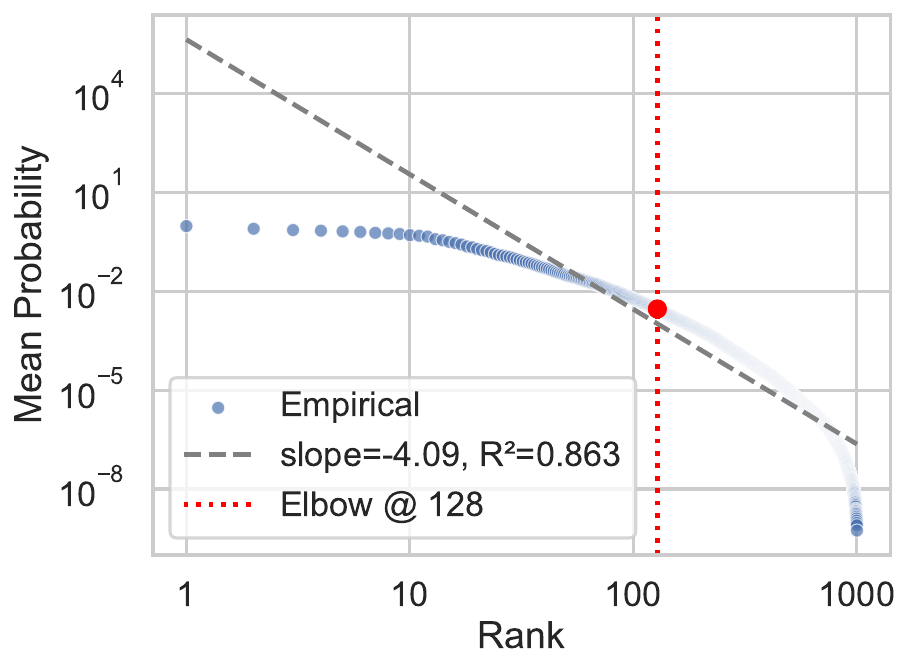}
    \caption*{\small KL Divergence \\ Random}
  \end{subfigure}%
  \begin{subfigure}[t]{0.24\textwidth}
    \centering
    \includegraphics[width=\linewidth]{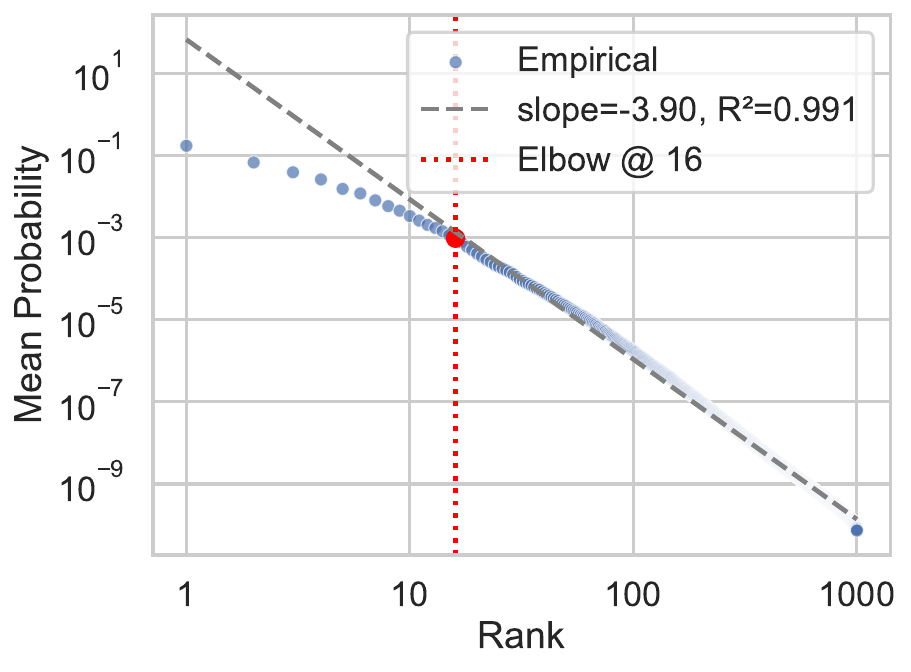}
    \caption*{\small KL Divergence \\ BM25}
  \end{subfigure}%
  \begin{subfigure}[t]{0.24\textwidth}
    \centering
    \includegraphics[width=\linewidth]{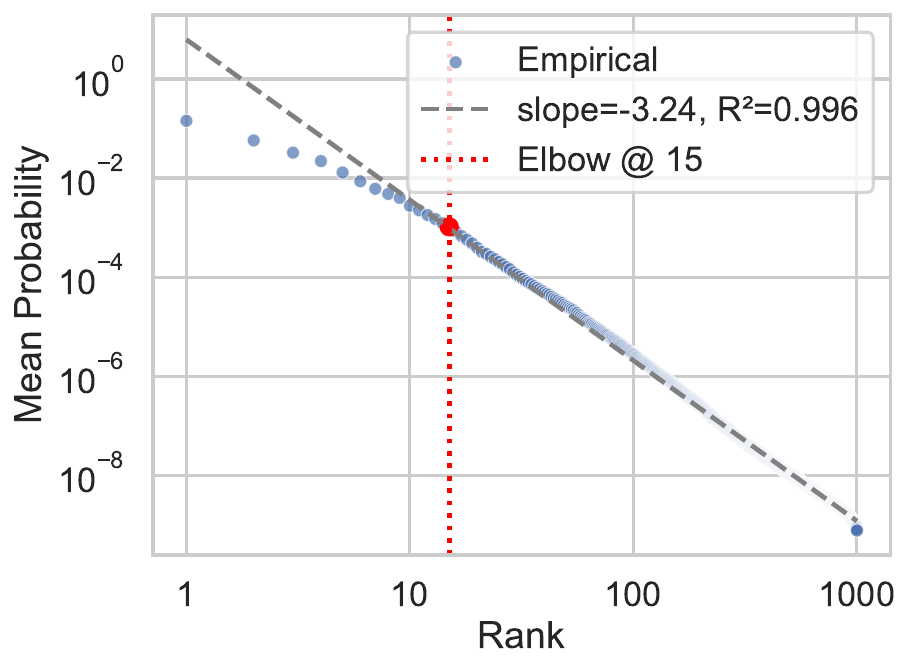}
    \caption*{\small KL Divergence \\ Cross-Encoder}
  \end{subfigure}%
  \begin{subfigure}[t]{0.24\textwidth}
    \centering
    \includegraphics[width=\linewidth]{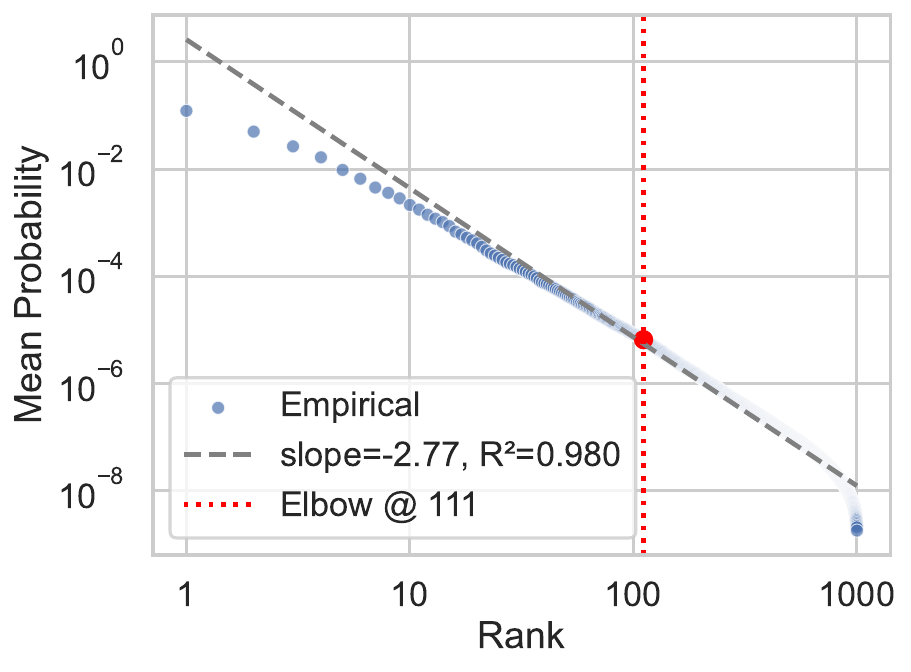}
    \caption*{\small KL Divergence \\ Ensemble}
  \end{subfigure}

  \medskip

  \begin{subfigure}[t]{0.24\textwidth}
    \centering
    \includegraphics[width=\linewidth]{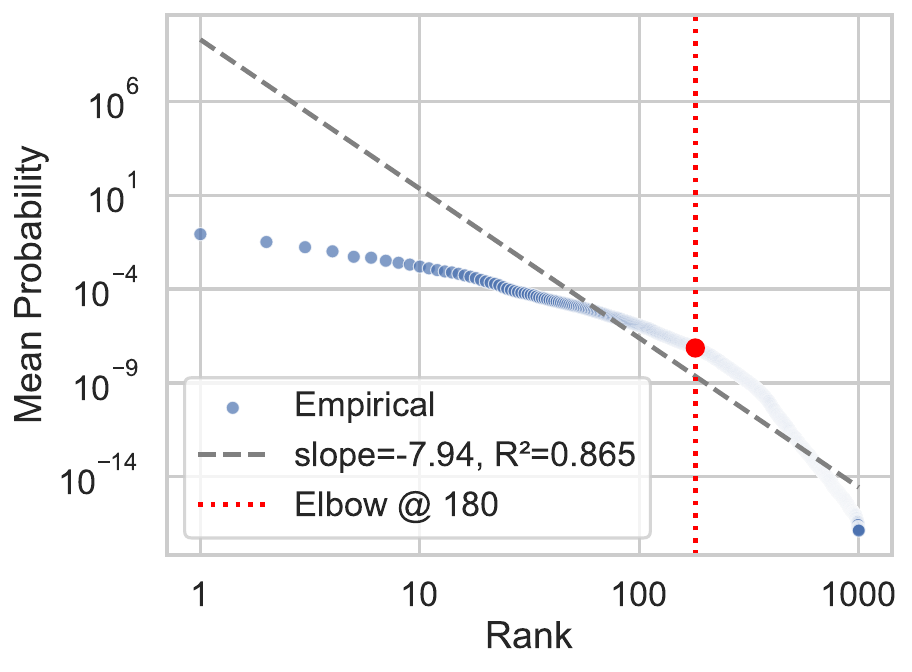}
    \caption*{\small MarginMSE \\ Random}
  \end{subfigure}%
  \begin{subfigure}[t]{0.24\textwidth}
    \centering
    \includegraphics[width=\linewidth]{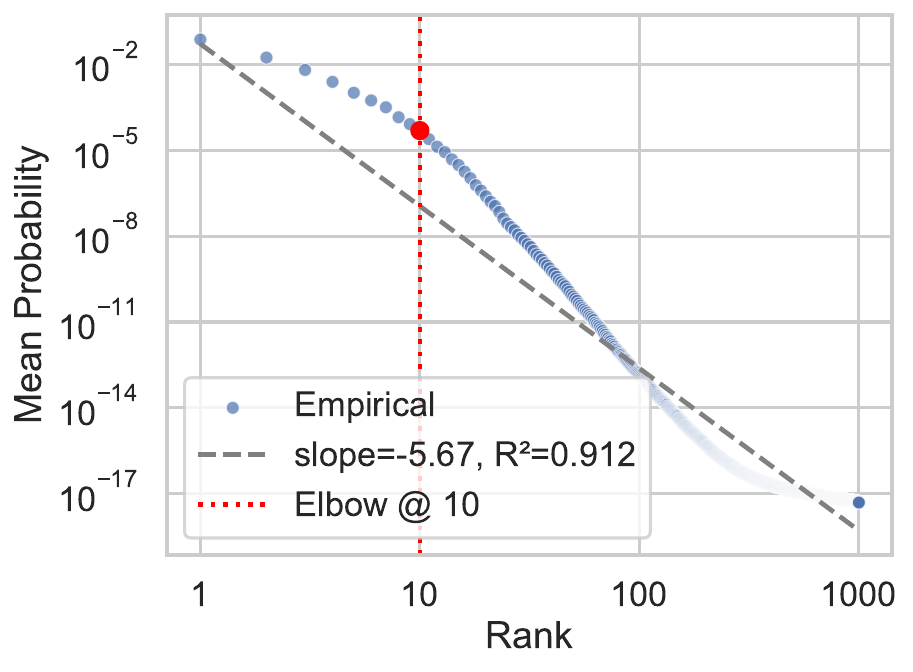}
    \caption*{\small MarginMSE \\ BM25}
  \end{subfigure}%
  \begin{subfigure}[t]{0.24\textwidth}
    \centering
    \includegraphics[width=\linewidth]{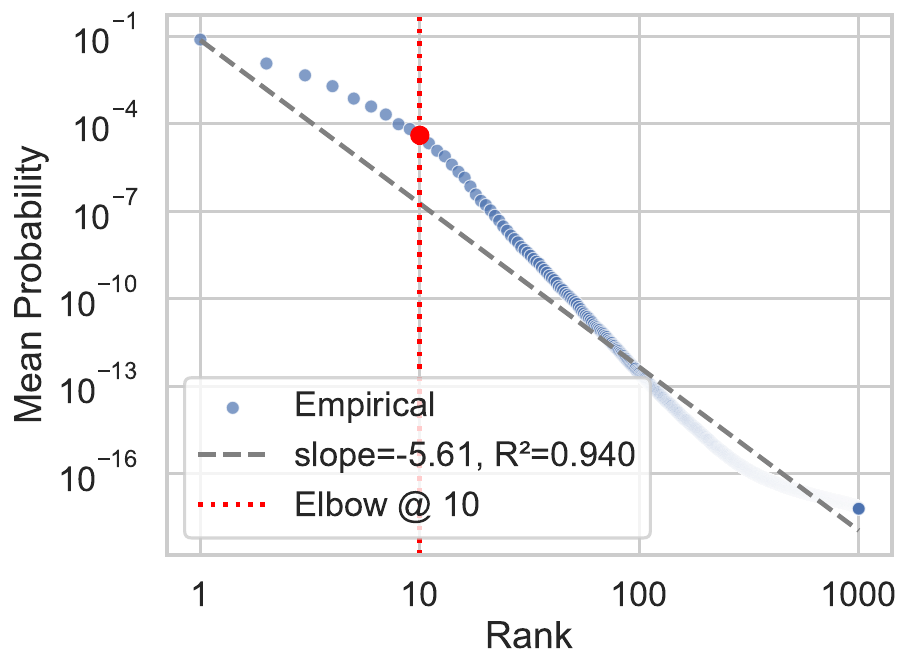}
    \caption*{\small MarginMSE \\ Cross-Encoder}
  \end{subfigure}%
  \begin{subfigure}[t]{0.24\textwidth}
    \centering
    \includegraphics[width=\linewidth]{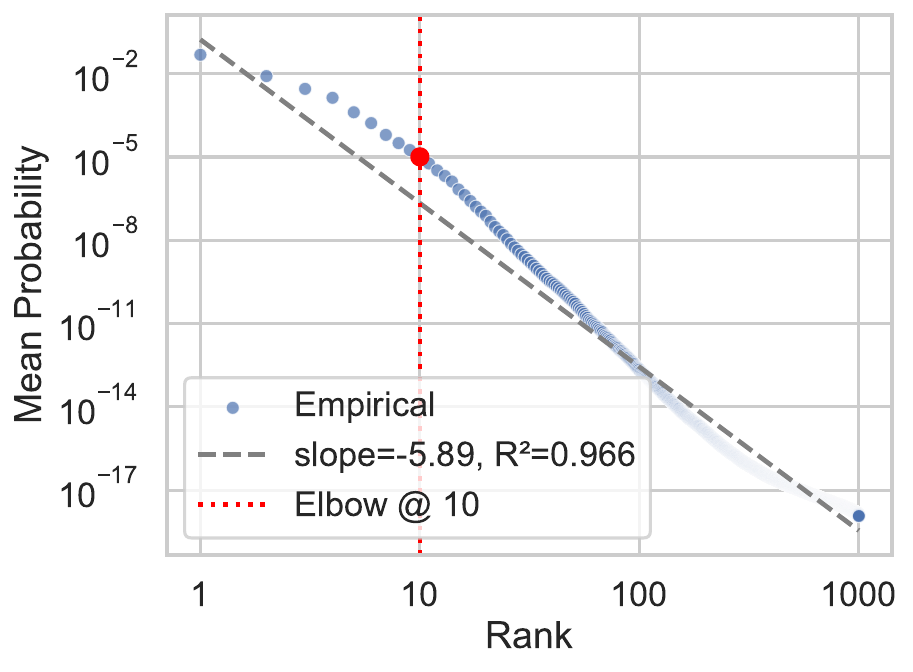}
    \caption*{\small MarginMSE \\ Ensemble}
  \end{subfigure}

  \medskip

  \begin{subfigure}[t]{0.24\textwidth}
    \centering
    \includegraphics[width=\linewidth]{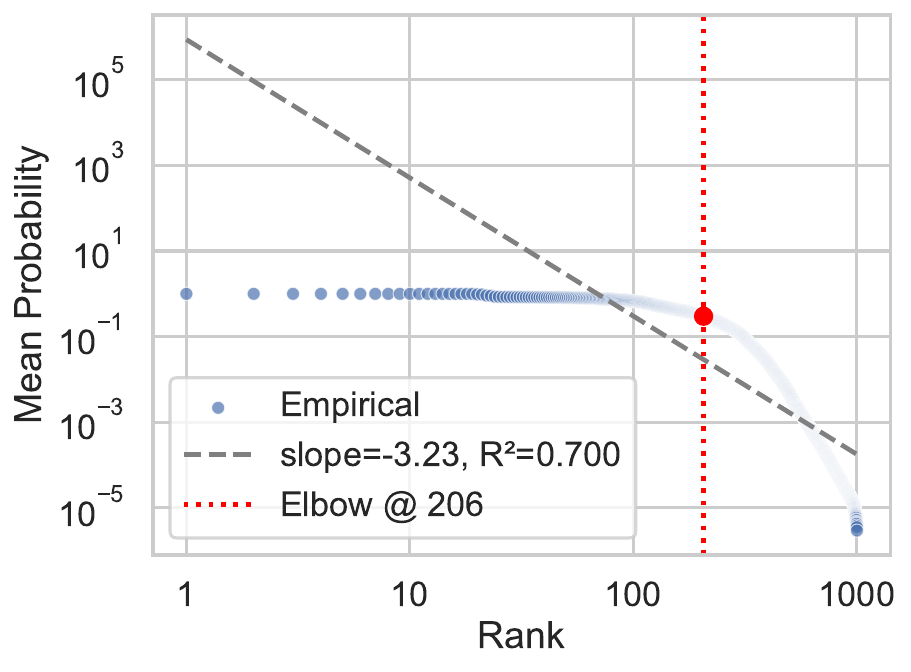}
    \caption*{\small RankNet \\ Random}
  \end{subfigure}%
  \begin{subfigure}[t]{0.24\textwidth}
    \centering
    \includegraphics[width=\linewidth]{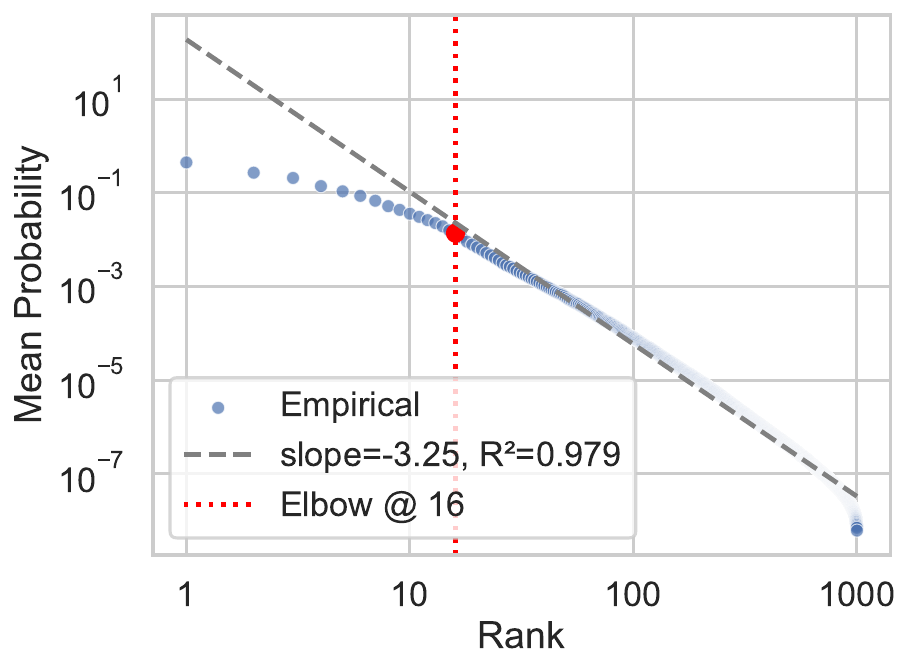}
    \caption*{\small RankNet \\ BM25}
  \end{subfigure}%
  \begin{subfigure}[t]{0.24\textwidth}
    \centering
    \includegraphics[width=\linewidth]{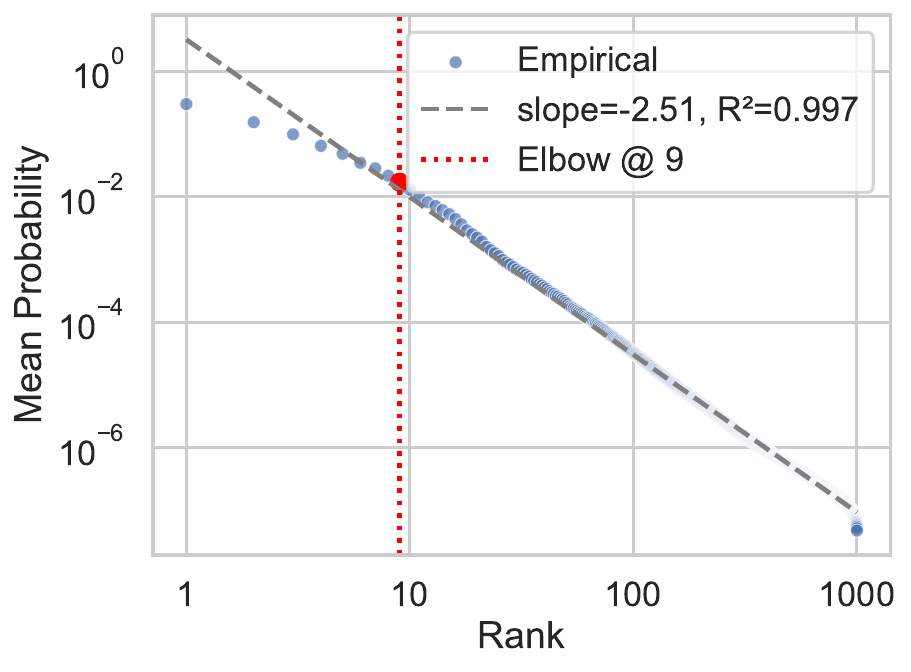}
    \caption*{\small RankNet \\ Cross-Encoder}
  \end{subfigure}%
  \begin{subfigure}[t]{0.24\textwidth}
    \centering
    \includegraphics[width=\linewidth]{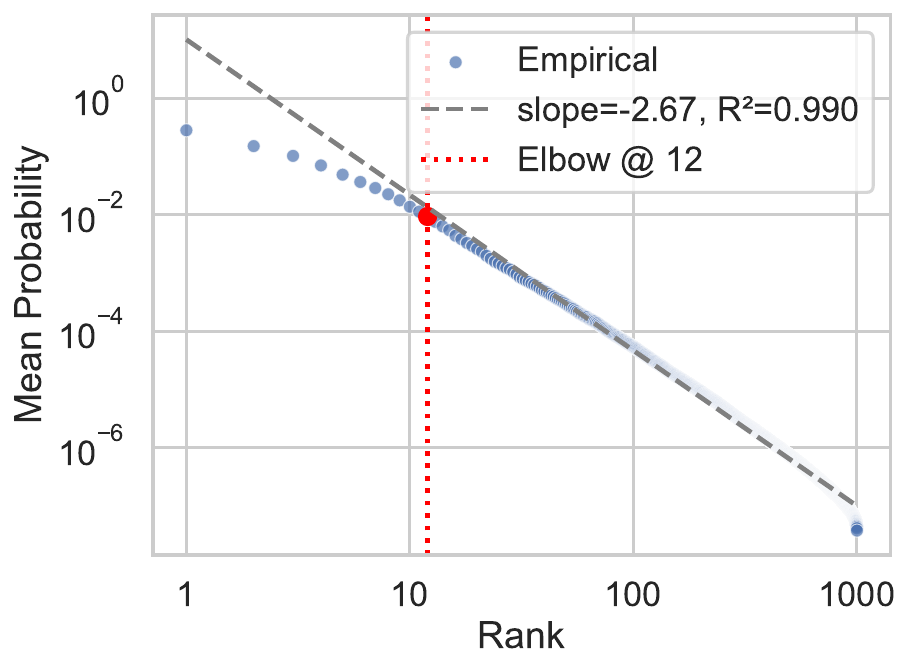}
    \caption*{\small RankNet \\ Ensemble}
  \end{subfigure}

  \caption{log-log plots of average score at each document rank on MS MARCO passage (TREC DL 2019 judged) for each loss function (rows) and domain (columns) when training a cross-encoder.}
  \label{fig:zipf_all_losses_domains}
\end{figure}

As a qualitative way to discriminate between domains, observe across Figure \ref{fig:zipf_all_losses_domains} the increasing alignment of model scores to a power law as harder negatives are applied. Under increasingly difficult negatives, the model’s score distribution stretches into a heavy‐tailed, near–power‐law form: only a few distractors receive high scores, while other documents are driven sharply downward. The slope of this tail offers a simple, domain‐agnostic measure of how confidently the model assigns relevance.  Importantly, when evaluated out of domain, this power‐law alignment vanishes entirely, potentially indicating overfitting, as indicated by reduced effectiveness out of domain in our main findings.

\section{Licenses}

\noindent \textbf{Datasets}
MSMARCO is licensed under the MIT license, strictly for non-commercial research purposes. NQ and DBPedia are provided under the CC BY-SA 3.0 license. ArguAna and Touché-2020 are provided under the CC BY 4.0 license. CQADupStack is provided under the Apache License 2.0 license. SciFact is provided under the CC BY-NC 2.0 license. SCIDOCS is provided under the GNU General Public License v3.0 license. HotpotQA is provided under the CC BY-SA 4.0 license. TREC-Covid test queries and judgements are provided under open domain however the underlying CORD-19 collection has variable licensing and we point readers to this \href{https://www.kaggle.com/datasets/allen-institute-for-ai/CORD-19-research-challenge/data?select=metadata.csv}{metadata} for more details.

\noindent \textbf{Models}
All base checkpoints (BERT and ELECTRA) are provided under Apache-2.0.

\end{document}